\documentclass[11pt]{article}

\usepackage[utf8]{inputenc} 
\usepackage[T1]{fontenc}    
\usepackage{hyperref}       
\usepackage{url}            
\usepackage{booktabs}       
\usepackage{amsfonts}       
\usepackage{nicefrac}       
\usepackage{microtype}      
\usepackage{graphicx}
\usepackage{amsmath}
\usepackage[top=3cm, right=3cm, left=3cm, bottom=4cm]{geometry}

\usepackage{bbm}
\usepackage{natbib}
\usepackage{subfigure}
\usepackage[justification=centerlast, font=small]{caption}
\usepackage{xcolor}
\usepackage{enumitem}
\usepackage{amssymb}
\usepackage{comment}
\usepackage{amsthm}
\usepackage{titling}

\usepackage[plain,noend]{algorithm2e}

\def\T{{ \mathrm{\scriptscriptstyle T} }}

\newtheorem{theorem}{Theorem}
\newtheorem{corollary}{Corollary}
\newtheorem{lemma}{Lemma}
\newtheorem{proposition}{Proposition}
\newtheorem{assumption}{Assumption}
\theoremstyle{definition}
\newtheorem{definition}{Definition}

\graphicspath{{art/}}

\newcommand{\Step}[2]{
	
	\textit{Step}  S#1. \hspace{0.1em} #2 \newline}
\newcommand{\Substep}[3]{
	
	\textit{Substep}  S#1.#2. \hspace{0.1em} #3 \newline}

\title{Generalized infinite factorization models}

\author{
	Lorenzo Schiavon\\
	Department of Statistical Sciences, University of Padova,\\
	Via Cesare Battisti 241, 35121 Padova, Italy,\\
	\texttt{lorenzo.schiavon@phd.unipd.it}
	\and
	Antonio Canale\\
	Department of Statistical Sciences, University of Padova,\\
	Via Cesare Battisti 241, 35121 Padova, Italy,\\
	\texttt{canale@stat.unipd.it}
	\and
	and David B. Dunson\\
	Department of Statistical Science, Duke University,\\
	Durham, North Carolina 27708, U.S.A.\\
	\texttt{dunson@duke.edu}
}

\begin{document}

\maketitle

\begin{abstract}
Factorization models express a statistical object of interest in terms of a collection of simpler objects.  For example, a matrix or tensor can be expressed as a sum of rank-one components. However, in practice, it can be challenging to infer the relative impact of the different components as well as the number of components.  A popular idea is to include infinitely many components having impact decreasing with the component index.  This article is motivated by two limitations of existing methods: (1) lack of careful consideration of the within component sparsity structure; and (2) no accommodation for grouped variables and other non-exchangeable structures. We propose a general class of infinite factorization models that address these limitations.  
Theoretical support is provided, practical gains are shown in simulation studies, and an ecology application focusing on modelling bird species occurrence is discussed.
\end{abstract}

{\small Adaptive Gibbs sampling; Bird species; Ecology; Factor analysis; High-dimensional data; Increasing shrinkage; Structured shrinkage.}

\section{Introduction}
\label{sec:intro}
Factorization models are used routinely to express matrices, tensors or other statistical objects based on simple components.  The likelihood for data $y$ under a general class of factorization models can be expressed as $L(y;\Lambda,\Psi, \Sigma)$, with $\Lambda = \{ \Lambda_h, h=1,\ldots,k\}$ a $p \times k$ matrix, $\Lambda_h = (\lambda_{1h},\ldots,\lambda_{ph})^\T$ 
the $h$th column vector of $\Lambda$, $\Psi$ and $\Sigma$ additional parameters, and $k$ a positive integer.  This class includes Gaussian linear factor models \citep{roweis1999}, exponential family factor models \citep{jun2013}, Gaussian copula factor models \citep{murray2013}, latent factor linear mixed models \citep{an2013}, probabilistic matrix factorization \citep{mnih2008},  underlying Gaussian factor models for mixed scale data \citep{reich2010}, and functional data factor models \citep{montagna2012}.  A fundamental problem is how to choose weights for the components and the number of components $k$.  This article proposes a general class of Bayesian methods to address this problem.

Although there is a rich literature, selection of $k$ is far from a solved problem.  
In unsupervised settings, it is common to fit the model for different choices of $k$ and then choose the value with the best goodness-of-fit criteria.  For likelihood models, the Bayesian information criteria is particularly popular.  It is also common to use an informal elbow rule, selecting the smallest $k$ such that the criteria improves only a small amount for $k+1$.  In specific contexts, formal model selection methods have been developed.  For example, 
taking a Bayesian approach, one can choose a prior for $k$ and attempt to approximate the posterior distribution of $k$ using Markov chain Monte Carlo; see \cite{lopes2004} for linear factor models, 
\cite{miller2018} for mixture models and 
\cite{yang2018} for matrix factorization.  Although such methods are conceptually appealing, computation can be prohibitive outside of specialized settings.

Due to these challenges it has become popular to rely on over-fitted factorization models, which include more than enough components with shrinkage priors adaptively removing unnecessary ones by shrinking their coefficients close to zero. Such approaches were proposed by \citet{rousseau2011asymptotic} for mixture models and 
\citet{Bhattacharya2011} for Gaussian linear factor models.  The latter approach specifically assumes an increasing shrinkage prior on the columns of the factor loadings matrix $\Lambda$.  \citet{Legramanti2020} recently modified this approach using a spike and slab structure \citep{mitchell88} that increases the mass on the spike for later columns.  

Although over-fitted factorizations are widely used, there are two key gaps in the literature.  The first is a careful development of the shrinkage properties of increasing shrinkage priors \citep{Durante2017}.  Outside of the factorization context and mostly motivated by high-dimensional regression, there is a rich literature recommending specific desirable properties for shrinkage priors.  These include high concentration at zero to favor shrinkage of small coefficients and heavy tails to avoid over shrinking large coefficients.  Motivated by this thinking, popular shrinkage priors have been developed including the Dirichlet-Laplace \citep{Bhattacharya2015} and horseshoe \citep{Carvalho2010}. Current increasing shrinkage priors, such as \citet{Bhattacharya2011}, were not designed to have the desirable shrinkage properties of these priors.  For this reason, ad hoc truncation and use of horseshoe/Dirichlet-Laplace can outperform increasing shrinkage priors in some contexts; for example, this was the case in \citet{ferrari2020}.

A second gap in the literature on over-fitted factorization priors is the lack of structured shrinkage.  The focus has been on priors for $\Lambda$ 
that are exchangeable within columns, with the level of shrinkage increasing with the column index. However, it is common in practice to have \textit{meta covariates} encoding features of the rows of $\Lambda$.  For example, the rows may correspond to different genes in genomic applications or species in ecology.  There is a rich literature on incorporating gene ontology  in statistical analyses of genomic data; refer, for example to \citet{thomas2009}.  In ecology it is common to include species traits in species distribution models \citep{ovaskainen2020}.  Beyond the Bayesian literature, 
it is common to include structured penalties, with the grouped Lasso \citep{yuan2006} a notable example.

Motivated by these deficiencies of current factorizations priors, this article proposes a broad class of generalized infinite factorization priors, along with corresponding theory and algorithms for routine Bayesian implementation.

\section{Generalized infinite factor models}

\subsection{Model specification}
\label{subsec:gif-model}
Suppose that an $n \times p$ data matrix $y$ is available. In our motivating application, $y_{ij}$ is a binary indicator of occurrence of bird species $j$ ($j=1,\ldots,p$) in sample $i$ ($i=1,\ldots,n$).
Consider the following general class of models, \label{rev:structure}
\begin{align}
\label{eq:gif}
&y_{ij} = t_j(z_{ij}), \quad z_{i}= \Lambda \eta_{i} + \epsilon_{i}, \quad \epsilon_{i} \sim f_{\epsilon},
\end{align}
with $\Lambda$ a $p \times k$ loadings matrix, $\eta_i$ a $k$ dimensional factor with diagonal covariance matrix $\Psi=\text{diag}(\psi_{11},\ldots,\psi_{kk})$, $\epsilon_i$ a $p$-dimensional error term independent of $\eta_i$, and
\label{rev:t-def}the function $t_j: \Re \to \Re$, for $j=1\dots, p$.
We refer to this class as generalized factorization models. Class \eqref{eq:gif} includes most of the cases mentioned in Section~\ref{sec:intro}.
When $\epsilon_i$ and $\eta_i$ are Gaussian random vectors and $t_j$ is the identity function, model \eqref{eq:gif} is a Gaussian linear factor model. 
With similar assumptions for $\epsilon_i$ and $\eta_i$, and assuming $t_j =F_j^{-1}(\Phi(z_{ij}))$, with $\Phi(z_{ij})$ the Gaussian cumulative distribution function, model \eqref{eq:gif} is a Gaussian copula factor model \citep{murray2013}. 
Exponential family factor models \citep{jun2013}, probabilistic matrix factorization \citep{mnih2008} and underlying Gaussian models for mixed scale data \citep{reich2010} can be obtained by appropriately defining the elements in \eqref{eq:gif}, whereas multivariate response regression models belong to this framework when $\eta_{i}$ is known. 

The matrix $\Omega =\text{var}(z_{i})$ can be expressed as 
$\Omega = \Lambda \Psi \Lambda^T + \Sigma,$
where $\Sigma=\text{var}(\epsilon_i)$. 
Following common practice in Bayesian factor analysis \citep{Bhattacharya2011}, we avoid imposing identifiability constraints on $\Lambda$ and assume $\Psi$ is pre-specified.  Our focus is on a new class of generalized infinite factor models induced through a novel class of priors for $\Lambda$ that allows infinitely many factors, $k=\infty$.  In particular, we let 
\begin{equation}
\label{eq:model}
\lambda_{jh}\mid \theta_{jh} \sim N(0, \theta_{jh}), \quad
\theta_{jh}= \tau_{0} \,\gamma_h \phi_{jh}, \quad 
\tau_{0} \sim f_{\tau_0}, \quad
\gamma_h\sim f_{\gamma_h}, \quad 
\phi_{jh}\sim f_{\phi_{j}},
\end{equation}
where  $f_{\tau_0}$, $f_{\gamma_h}$, and $f_{\phi_{j}}$ are supported on $[0,\infty)$ with positive probability mass on $(0,\infty)$.  
The local $\phi_{jh}$, column-specific $\gamma_h$, and global $\tau_0$ scales are all independent \textit{a priori}.
We let $N(0,0)$ denote a degenerate distribution with all its mass at zero.
Expression \eqref{eq:model} induces a class of scale-mixture of Gaussian shrinkage priors \citep{Polson2010} for the loadings. 
Although we allow infinitely many columns in $\Lambda$, 
\eqref{eq:model} induces a prior for 
$\Omega$ supported on the set of $p \times p$ positive semi-definite matrices under mild conditions reported in Proposition S1 in the Supplementary Material.

Differently from most of the existing literature on shrinkage priors, we want to define a non-exchangeable structure that includes \textit{meta covariates} $x$ informing the sparsity structure of $\Lambda$.  In our context, meta covariates provide information to distinguish the $p$ different variables as opposed to traditional covariates that serve to distinguish the $n$ subjects. Letting 
$x$ denote a $p \times q$ matrix of such meta covariates, we choose $f_{\phi_j}$ not depending on the index $h$ and such that 
\begin{equation}
\label{eq:phi}
E(\phi_{jh} \mid \beta_h)=g(x_j^T \beta_h), \quad \beta_h=(\beta_{1h},\ldots, \beta_{qh})^\T, \quad	\beta_{mh} \sim f_{\beta} \quad (m=1,\ldots,q)
\end{equation}
where $g: \Re \to {  \cal A} \subset \Re_+$ is a \label{rev:g-def}known smooth one-to-one differentiable link function,
$x_j = (x_{j1},\ldots,x_{jq})^T$ denotes the $j$th row vector of $x$, and $\beta_h$ are coefficients controlling the impact of the meta covariates on shrinkage of the factor loadings in the $h$th column of $\Lambda$.

To illustrate the usefulness of \eqref{eq:phi}, consider the previously introduced ecological study and suppose \label{rev:group-ex} $x_j = \{1,\mathbbm{1}{(\kappa_j=2)},\ldots,\mathbbm{1}{(\kappa_j=q)}\}^T$, where $\kappa_j \in \{1,\ldots,q\}$ denotes the phylogenetic order of species $j$.  
Species of the same order may tend to have similarities that can be expressed in terms of a shared pattern of high or low loadings on the same latent factors.
\begin{figure}[t]
	\centering
	\includegraphics[width=.5\textwidth]{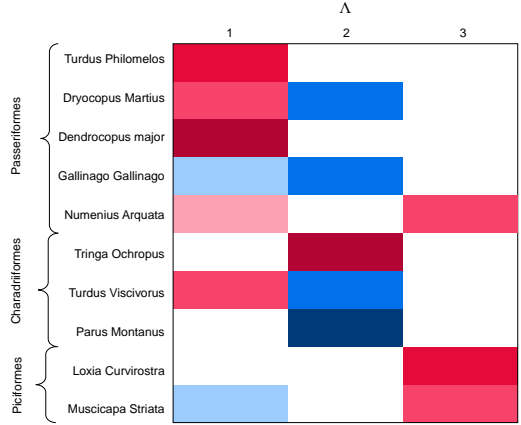} 
	\caption{Illustrative loadings matrix of an ecology application, where the rows refer to ten bird species belonging to three phylogenetic orders. White cells represent the elements of $\Lambda$ equal to zero, while blue and red cells represent negative and positive values, respectively.}
	\label{fig:toyex}
\end{figure}
To illustrate this situation, we simulate a loadings matrix, displayed in Fig.~\ref{fig:toyex}, \label{rev:fig1} sampling from the prior introduced in Section 3 where $\text{pr}(\lambda_{jh}=0)>\text{pr}(\phi_{jh}=0)>0$.
The loadings within each column are penalized basing on the group structure identified by the $q=3$ phylogenetic orders (Passeriformes, Charadriiformes, and Piciformes) of the $p=10$ birds species considered. 
Our proposed prior allows for the possibility of such structure while not imposing it.
In the bird ecology application, $x$ can be defined to include not just phylogenetic placement of each bird species but also species traits, such as size or diet \citep{tikhonov2020}. Related meta covariates are widely available, both in other ecology applications \citep{miller2019} and in other fields such as genomics \citep{thomas2009}. 

\subsection{Properties}
\label{subsec:gif-prop}

In this section we present some properties motivating the shrinkage process in \eqref{eq:model} and provide insight into prior elicitation.  It is important to relate the choice of hyperparameters to the signal-to-noise ratio, expressed as the proportion of variance explained by the factors. \label{rev:prior-prop} Section S2.4 of the Supplementary Material provides a study of the posterior distribution of the proportion of variance explained; the posterior tends to be robust to hyperparameter choice.  Below we study key properties of our prior, including 
an increasing shrinkage property, the ability of the induced marginal prior to accommodate both sparse and large signals, and control of the multiplicity problem in sparse settings.  Proofs are included in the Appendix and in Section S1 of Supplementary Material.
\label{rev:prior-elicitation}This theory illuminates the role of hyperparameters; specific recommendations of hyperparameter choice in practice are illustrated under the model settings of Section \ref{subsec:sis-model}.

To formalize the increasing shrinkage property, we introduce the following  definition.
\begin{definition}
	Letting $\Pi_\Lambda$ denote a shrinkage prior\label{rev:shrinkage} on $\Lambda$,
	$\Pi_\Lambda$ is a weakly increasing shrinkage prior if 
	$\text{\normalfont var}(\lambda_{j (h-1)})>
	\text{\normalfont var}(\lambda_{j h})$ for $j$ in $1,\ldots, p$ and $h=2,\ldots, \infty$. $\Pi_\Lambda$ is a strongly increasing shrinkage prior if 
	$\text{\normalfont var}(\lambda_{s (h-1)}) >
	\text{\normalfont var}(\lambda_{j h})$, for $j,s$ in $\{1,\ldots, p\}$ and $h=2,\ldots, \infty$. 
	\label{def:shrink}
\end{definition}
\vskip -12pt

Weakly increasing shrinkage corresponds to the prior variance increasing across columns within each row of $\Lambda$, while strongly increasing shrinkage implies that the prior variance of any loading element is larger than all elements with a higher column index.  
In the following Theorem, we show that the process in \eqref{eq:model} induces weakly increasing shrinkage under a simple sufficient condition.  

\begin{theorem}
	\label{th:var-dec}
	Expression \eqref{eq:model} is a weakly increasing shrinkage prior under Definition~\ref{def:shrink} if 
	$\mbox{E}(\gamma_h)>\mbox{E}(\gamma_{h+1})$ for any $h$.
\end{theorem}

Increasing shrinkage priors favor a decreasing contribution of higher indexed columns of $\Lambda$ to the covariance $\Omega$. In addition to inducing a flexible shrinkage structure that allows different factors to have a different sparsity structure in their loadings, this allows one to accurately approximate the
likelihood $L(y;\Lambda,\Psi, \Sigma)$
by $L(y;\Lambda_{H},\Psi_{H},\Sigma)$,
with $\Lambda_H$ containing the first $H$ columns of the infinite matrix $\Lambda$ and $\Psi_H$ the first $H$ rows and columns of $\Psi$.
To measure the induced truncation error of $\Omega_H=\Lambda_H \Psi_H \Lambda_H^\T  +\Sigma$, we use the trace of $\Omega$.
The trace is justified by the fact that the maximum error occurring in an element of $\Omega$ due to truncation always lies along the diagonal and \label{rev:nuclear-norm} by the relation between difference of traces and the nuclear norm, routinely used to approximate low rank minimization problems \citep{liu2010}. 
The following Proposition provides conditions on prior \eqref{eq:model} so that the under-estimation of $\Omega$ that occurs by truncating decreases exponentially fast as $H$ increases.
\begin{proposition}
	\label{pr:upp-bound}
	Let $E(\tau_0)$ and $E(\phi_{jh})$ be finite for $j=1, \ldots, p$ and $h=1,\ldots, \infty$ and $E(\gamma_h) = a b^{h-1}$ with $a>0$ and $b \in (0,1)$ for all  $h=1,\ldots, \infty$. Let $c>0$ be a sufficiently large number such that $c \geq \max_{h=1, \ldots, \infty} \psi_{hh}$. If 
	$$
	m_{\Omega} = \min_{j=1,\ldots,p} \left[E(\sigma_j^{-2}), \, E\left\{\left(\sum_{h=1}^{\infty} \psi_{hh} \lambda_{jh}^2\right)^{-1}\right\}\right]< \infty,
	$$
	then for any $T \in (0,1)$, 
	\label{rev:OmegaH-1}
	\begin{equation*}
	\text{\normalfont pr}\left\{\frac{ \text{\normalfont tr}(\Omega_H)}{\text{\normalfont tr}({\Omega})}\leq T\right\}\leq \bigg(\frac{1}{1-T}\bigg) \,  ac\frac{b^H}{1-b}\,  m_{\Omega}\, E(\tau_0)\, \sum_{j=1}^{p}  E(\phi_{j 1} ) .
	\end{equation*}
	\label{theo:1}
\end{proposition}
\vskip -12pt 

The above increasing shrinkage properties can be satisfied by naive priors that over-shrink the elements of $\Lambda$.  It is important to avoid such over-shrinkage and allow not only many elements that are $\approx 0$ but also a small proportion of large coefficients.  A similar motivation applies in the literature on shrinkage priors in regression \citep{Carvalho2010}. Borrowing from that literature, the marginal prior for $\lambda_{jh}$ should be concentrated at zero to reduce mean square error by shrinking small coefficients to zero but with heavy tails to avoid over-shrinking the signal.

To quantify the prior concentration of \eqref{eq:model} in an $\epsilon$ neighbourhood of zero, we can obtain 
\begin{equation}
\label{eq:prior-conc}
\text{pr}(|\lambda_{jh}|>\epsilon) \leq \frac{E(\tau_0) \,E(\gamma_{h})\, E(\phi_{jh})}{\epsilon^2}
\end{equation}	
as a consequence of Markov's inequality. 
Common practice in local-global shrinkage priors chooses $\mbox{E}(\tau_0)$ small while assigning a heavy-tailed density to the local or column scales. In our case,
\eqref{eq:phi} allows the bound in \eqref{eq:prior-conc} to be regulated by meta covariates $x$, while, under the condition in Theorem~\ref{th:var-dec}, decreasing $E(\gamma_h)$ with column index causes an increasing concentration near zero, since $E(\phi_{jh})=E(\phi_{jl})$ for every $h, l \in \{1,\ldots, \infty\}$. The means of the column and the local scales control prior concentration near zero, while over-shrinkage can be ameliorated by choosing $f_{\phi_j}$ or $f_{\gamma_h}$ ($h=1,\ldots,\infty$) heavy tailed.
The following Proposition provides a condition on the prior to guarantee a heavy tailed marginal distribution for $\lambda_{jh}$.  A random variable has power law tails if its cumulative distribution function $F$ has $1-F(t) \geq c t^{-\alpha }$ for constants $c>0$, $\alpha>0$, and for any $t>L$ for $L$ sufficiently large.
\begin{proposition}
	\label{pr:power-law}
	If at least one scale parameter among $\tau_0$, $\gamma_h$ or $\phi_{jh}$ is characterized by a power law tail prior distribution, then the prior marginal distribution of $\lambda_{jh}$ has power law tails.
\end{proposition}
An important consequence of the heavy tailed property is avoidance of over-shrinkage of large signals.
This is often formalized via a tail robustness property \citep{Carvalho2010}.  As an initial result, key to showing sufficient conditions for a type of local tail robustness, we provide the following Lemma on the derivative of the \label{rev:lemma1}log prior in the limit as the value of $\lambda_{jh} \to \infty.$   
\begin{lemma}
	\label{lem:score}
	If at least one scale parameter among $\tau_0$, $\gamma_h$ or $\phi_{jh}$ has a prior with power law tails for 
	any possible prior distribution of $\beta_h$,  then for any finite truncation level $H$,
	\begin{equation*}
	\lim_{\lambda \rightarrow \infty } \frac{\partial \log \{f_{\lambda_{jh}\mid \Lambda_{-jh} }(\lambda)\}}{\partial \lambda} = 0
	\end{equation*}
	where $f_{\lambda_{jh}\mid \Lambda_{-jh} }(\lambda)$ is the conditional distribution of $\lambda_{jh}$ given the other elements of $\Lambda_H$.
\end{lemma}
The following definition introduces a type of local tail robustness property. 
\begin{definition}
	\label{def:robustness}
	Consider model \eqref{eq:gif} with factors $\eta$ known.
	Let $f_{\lambda_{jh}\mid y, \eta, \Lambda_{-jh} }(\lambda)$ denote the posterior density of $\lambda_{jh}$, given the data, conditional on any possible value of the other elements of 
	$\Lambda_H$ for any finite $H$, and let $\hat{\lambda}_{jh}$ denote the conditional maximum likelihood estimate of $\lambda_{jh}$ for any possible value of the other elements of $\Lambda_H$. We say that the prior on $\lambda_{jh}$ is tail robust if
	\label{rev:flambda}
	\begin{equation*}
	\lim_{\hat{\lambda}_{jh} \to \infty} \left| \hat{\lambda}_{jh} - \arg \max_{\lambda} f_{\lambda_{jh}\mid y, \eta, \Lambda_{-jh}} (\lambda) \right|= 0.
	\end{equation*}
\end{definition}
For a given sample, $\hat{\lambda}_{jh}$ is a fixed quantity; the above limit should be interpreted as what happens as the data support a larger and larger maximum likelihood estimate.
In order for tail robustness to hold, 
we need the data to be sufficiently informative about the parameter $\lambda_{jh}$ and the likelihood to be  sufficiently regular; this is formalized as follows. 
\begin{assumption}
	\label{ass:fisher}
	Let $L(y; \Lambda, \eta, \Sigma)$ denote the likelihood for data $y$ conditionally on latent variables $\eta$, let  $l_s(\lambda)$ denote the derivative function of the log-likelihood with respect to $\lambda_{jh}$, and let $\mathcal{J}(\hat{\lambda}_{jh})$ denote the negative of the second derivative of the log-likelihood with respect to $\lambda_{jh}$, evaluated at the conditional maximum likelihood estimate $\hat{\lambda}_{jh}$.
	Then $l_s(\lambda)$ is a continuous function for every $\lambda \in \Re$ and $\mathcal{J}(\hat{\lambda}_{jh})\geq \nu(\hat{\lambda}_{jh})$, where $\nu(\hat{\lambda}_{jh})$ is of order $O(1)$ as $\hat{\lambda}_{jh} \to \infty$.  
\end{assumption}
This assumption can be verified for most of the cases mentioned in Section 1; for example, for Gaussian linear factor models 
$\mathcal{J}(\hat{\lambda}_{jh})$ is of order $O(1)$ with respect to $\hat{\lambda}_{jh}$.  

\begin{theorem}
	\label{th:robustness}
	Under Assumption~\ref{ass:fisher}, if at least one scale parameter among $\tau_0$, $\gamma_h$ or $\phi_{jh}$ is power law tail distributed for any possible prior distribution of $\beta_h$, then the prior on $\lambda_{jh}$ is tail robust under Definition~\ref{def:robustness}.
\end{theorem}

As an additional desirable property, we would like to control for the multiplicity problem within each column $\lambda_h$ of the loadings matrix, corresponding to increasing numbers of false signals as dimension $p$ increases. 
This can be accomplished by imposing an asymptotically increasingly sparse property on the prior, which is defined as follows. 
\begin{definition}
	\label{def:col-prior}
	Let $|\text{\normalfont supp}_{\epsilon}(\lambda_{h})|$ denote the cardinality of $\text{\normalfont supp}_{\epsilon}(\lambda_{h}) = (j : |\lambda_{jh}| > \epsilon)$.
	Let $s_p=o(p)$ 
	such that $s_p \geq c_s \log(p)/p$ for some constant $c_s>0$.
	We say that the prior on $\Lambda$ defined in~\eqref{eq:model} is an asymptotically increasingly sparse prior if 
	\begin{equation*}
	\lim_{p \to \infty}	\text{\normalfont pr}\{|\text{\normalfont supp}_{\epsilon}(\lambda_{h})|> a \, s_p \mid \gamma_h, \tau_0\} =0,\qquad \mbox{for some constant $a>0$}.
	\end{equation*}
\end{definition}
The quantity $|\text{supp}_{\epsilon}(\lambda_{h})|$ represents an approximate measure of model size for continuous shrinkage priors and, conditionally on $\beta_h$, $\gamma_h$, and $\tau_0$, it is \textit{a priori} distributed as a sum of independent Bernoulli random variables $\text{Ber}(\zeta_{\epsilon jh})$, where 
\begin{equation*}
\zeta_{\epsilon jh}=\text{pr}(|\lambda_{jh}|>\epsilon \mid \beta_h, \gamma_h, \tau_0)\leq \frac{\tau_0 \,\gamma_{h}\, g(x_{j}^\T\beta_{h})}{\epsilon^2}.
\end{equation*}
We now provide sufficient conditions for an asymptotically increasingly sparse prior, \label{rev:theorem3} allowing regulation of the sparsity behaviour of the prior of the columns of $\Lambda$ for increasing dimension $p$.
\begin{theorem}
	\label{th:asym-sparse}
	Consider prior \eqref{eq:model} with $\phi_{jh}$ ($j=1,\ldots,p$) \textsl{a priori} independent given $\beta_h$.
	If  $\text{\normalfont pr} \{ g(x_j^{\T}\beta_h)\leq \nu_{j}(p)\} = 1$, with
	$\nu_{j}(p)=O\{\log(p)/p\}$, ($j=1,\ldots,p)$, then the prior on $\Lambda$ is asymptotically increasingly sparse under Definition~\ref{def:col-prior}.
\end{theorem}
\label{rev:g-bounded} The condition of the theorem is easily satisfied, for example, if $g$ is the multiplication of a bounded function and a suitable offset depending on $p$ as assumed in Section \ref{subsec:sis-model}.
The multiplicative gamma process \citep{Bhattacharya2011} and cumulative shrinkage process \citep{Legramanti2020} do not satisfy the sufficient conditions of Theorem~\ref{th:asym-sparse}, and, furthermore, the following lemma holds. 
\begin{lemma}
	\label{lem:mgp-cusp-nai}
	The multiplicative gamma process prior \citep{Bhattacharya2011} and the cumulative shrinkage process prior \citep{Legramanti2020} are not asymptotically increasing sparse under Definition~\ref{def:col-prior}.
\end{lemma}
\label{rev:theorem3-post}
Although this Section has focused on properties of the prior, we find empirically that these properties tend to carry over to the posterior, as will be illustrated in the subsequent sections.  For example, the posterior exhibits 
asymptotic increasing sparsity; refer to Table 2 of Section~\ref{sec:sim}, which shows results for a novel process in our proposed class that is much more effective than current approaches in identifying the true sparsity structure, particularly when $p$ is large.

\section{Structured increasing shrinkage process}
\subsection{Model specification}
\label{subsec:sis-model}
In this section we propose a structured increasing shrinkage process prior for generalized infinite factor models satisfying all the sufficient conditions in Propositions~\ref{pr:upp-bound}--\ref{pr:power-law} and Theorems~\ref{th:robustness}--\ref{th:asym-sparse}.
Let $\text{Ga}(a,b)$ denote the gamma distribution with mean $a/b$ and variance $a/b^2$.
Following the notation of Section~\ref{subsec:gif-model}, we specify
\begin{align}
\label{eq:CUSPadj}
&\tau_0=1, \quad 
\gamma_h = \vartheta_{h}\rho_h, \quad 
\phi_{jh}\mid \beta_h \sim \text{Ber}\{\text{logit}^{-1}(x_j^{\T}\beta_h)\, c_p\}, \\
&\vartheta_{h}^{-1}\sim \text{Ga}(a_{\theta}, b_{\theta}),  \quad a_\theta >1, \quad  \rho_h=\text{Ber}\left(1-\pi_h\right), \quad
\beta_h \sim N_q(0, \sigma_\beta^2 I_q), \nonumber 
\end{align}
where \label{rev:g-sis} we assume the link $g(x)=\text{logit}^{-1}(x) c_p$, with $\text{logit}^{-1}(x)=e^{x}/(1+e^{x})$ and  $c_p \in  (0,1)$ a possible offset.
The parameter $\pi_h = \mbox{pr}(\gamma_h =0)$ follows a stick-breaking construction, 
\begin{equation*}
\pi_h= \sum_{l=1}^{h} w_l, \quad w_l=v_l\prod_{m=1}^{l-1}(1-v_m), \quad v_m \sim \text{Be}(1, \alpha),
\end{equation*}
with $\text{Be}(a, b)$ the beta distribution with mean $a/(a+b)$, such that $\pi_{h+1}>\pi_{h}$ is guaranteed for any $h=1,\ldots,\infty$ and $\lim_{h \rightarrow \infty} \pi_h= 1$ almost surely.
The prior expected number of non degenerate $\Lambda$ columns is $E(\sum_{h=1}^\infty \rho_h)=\alpha$ \citep{Legramanti2020}, suggesting setting $\alpha$ equal to the expected number of active factors. 
The prior specification is completed assuming  $\Sigma=\text{diag}(\sigma_1^2,\ldots, \sigma_p^2)$ with $\sigma_{j}^{-2} \sim \text{Ga}(a_\sigma, b_\sigma)$  for $j=1,\ldots,p$, consistently with the literature.  The hyperparameters can be chosen based on one's prior expectation of the signal-to-noise ratio, as $\sigma_j^2$ is the contribution of the noise component to the total variance of the $j$th variable. A sensitivity study in Section S2.4 of the Supplementary Material, however, shows that posterior distributions tend to be robust to the specification of $a_{\sigma},b_{\sigma}$.
\label{rev:prior-study} Regarding prior elicitation, we recommend setting $b_\theta \geq a_\theta$ to induce a high enough proportion of variance explained by the factor model. In Section S2.4 in the Supplementary Materials we report empirical evidence of the effect of different prior parameters on this quantity.

The above specification respects \eqref{eq:model} and, consequently, the following corollary holds.
\begin{corollary}
	\label{co:well-prior2}
	The structured increasing shrinkage process defined in (\ref{eq:CUSPadj})
	\begin{itemize}
		\item[$i.$]is a strongly increasing shrinkage prior under Definition  \ref{def:shrink};
		\item[$ii.$] for any $T \in (0,1)$,
		\label{rev:OmegaH-2}
		\begin{equation*}
		\text{\normalfont pr}\left\{\frac{ \text{\normalfont tr}(\Omega_H)}{\text{\normalfont tr}({\Omega})}\leq T\right\}\leq \bigg(\frac{1}{1-T}\bigg)\, \frac{b^H}{1-b} \, \theta_0\,\frac{a_\sigma}{b_\sigma} \,\sum_{j=1}^{p}E(\phi_{j1}),
		\end{equation*}
		with $b=\{\alpha(1+\alpha)\}^{-1}$ and $\theta_0 = E(\vartheta_h)$.
	\end{itemize}
\end{corollary}
\label{rev:post-trunc}
We conducted a simulation study on the posterior distribution of $\{\text{tr}(\Omega_H)/\text{tr}({\Omega})\leq T\}$ for varying hyperparameters,  and found that the results, reported in Section S2.4 of the Supplementary Material, were quite consistent with our prior truncation error bounds.

The prior concentration of the structured increasing shrinkage process in (\ref{eq:CUSPadj}) follows from (\ref{eq:prior-conc}): 
\begin{equation*}
\text{pr}(|\lambda_{jh}|>\epsilon) \leq 
\frac{E(\vartheta_h) \{1-E(\pi_h)\} E(\phi_{jh})}{\epsilon^2} =
\frac{\theta_{0}\,\{\alpha/(1+\alpha)\}^h}{\epsilon^2} \frac{c_p}{2}.
\end{equation*}	

In addition, the inverse gamma prior on $\vartheta_h$ implies a power law tail distribution on $\gamma_h$ inducing robustness properties on $\lambda_{jh}$ as formalized by the next corollary of Proposition~\ref{pr:power-law} and Theorem~\ref{th:robustness}. 
\begin{corollary}
	\label{co:robustness}
	Under the structured increasing shrinkage process defined in (\ref{eq:CUSPadj})
	\begin{itemize}
		\item[$i.$] the marginal prior distribution on $\lambda_{jh}$ ($j=1,\ldots,p$; $h=1,2, \ldots$) has power law tails; 
		\item[$ii.$] under Assumption~\ref{ass:fisher}, the prior on $\lambda_{jh}$ ($j=1,\ldots,p$; $h=1,2, \ldots$) is tail robust under Definition \ref{def:robustness}.
	\end{itemize}
	
\end{corollary}

Finally, it is important to assess the joint sparsity properties of the prior on each column of $\Lambda$. This is formalized in the following corollary of Theorem \ref{th:asym-sparse}.
\begin{corollary}
	\label{co:asymp-sparse}
	If $c_p = O\{\log(p)/p\}$ the structured increasing shrinkage process defined in (\ref{eq:CUSPadj}) is asymptotically increasingly sparse under Definition \ref{def:col-prior}.
\end{corollary}

\subsection{Posterior computations}
\label{subsec:post-comp}
Posterior inference is conducted via Markov chain Monte Carlo sampling.
Following common practice in infinite factor models \citep{Bhattacharya2011, Legramanti2020, schiavon2020} we use an adaptive Gibbs algorithm, which attempts to infer the best truncation level $H$ while drawing from the posterior distribution of the parameters. The value of $H$ is adapted only at some Gibbs iterations by discarding redundant factors and, if no redundant factors are identified, by adding a new factor by sampling its parameters from the prior distribution. 
Convergence of the Markov chain is guaranteed by satisfying the diminishing adaptation condition in Theorem 5 of \cite{roberts2007}, by specifying the probability of occurrence of an adaptive iteration $t$ as equal to $p(t)=\exp(\alpha_0 + \alpha_1 t)$, where $\alpha_0$ and $\alpha_1$ are negative constants, such that frequency of adaptation decreases.

The decomposition of $\gamma_h$ into two parameters $\rho_h$ and $\vartheta_h$ allows one to identify the inactive columns of $\Lambda$, corresponding to the redundant factors, as those with $\rho_h = 0$, while $H_a$ indicates the number of active columns of $\Lambda$. Consequently, at the adaptive iteration $t+1$, the truncation level $H$ is set to $H^{(t+1)}=H^{(t)}_a+1$ if $H^{(t)}_a<H^{(t)}-1$, and $H^{(t+1)}=H^{(t)}+1$ otherwise. 
Given $H^{(t+1)}$, the number of factors of the truncated model at iteration $t+1$, the sampler draws the model parameters from the corresponding posterior full conditional distributions.
The detailed steps of the adaptive Gibbs sampler for the structured increasing shrinkage prior in case of Gaussian or binary data are reported in the Supplementary Material, \label{rev:mixing} as well as trace plots of the posterior samples for some parameters of the model in Section 5 (see Section S3.2), showing good mixing.

\subsection{Identifiability and posterior summaries}
\label{subsec:identifiability}

Non-identifiability of the latent structure creates problems in interpretation of the results from Markov chain Monte Carlo samples.  Indeed, both $\Lambda$ and $\eta$ are only identifiable up to an arbitrary rotation $P$ with $PP^T=I_k$. This is a well known problem in Bayesian factor models and there is a rich literature proposing post-processing algorithms that 
align posterior samples $\Lambda^{(t)}$, so that one can then obtain interpretable posterior summaries.  Refer to \citet{mcparland2014}, \citet{assmann2016}, and \citet{roy2019} for alternative post-processing algorithms in related contexts.

Unfortunately, such post-hoc alignment algorithms destroy the structure we have carefully imposed on the loadings in terms of sparsity and dependence on meta covariates.  Therefore, we propose a different solution to obtain a point estimate of $\Lambda$ based on finding a representative Monte Carlo draw $\Lambda^{(t)}$ consistently with the proposals of \citet{dahl2006} and \citet{wade} in the context of Bayesian model-based clustering. Specifically, we summarize $\Lambda$ and $\beta = (\beta_1, \beta_2, \dots )$ through  $\Lambda^{(t^*)}$ and $\beta^{(t^*)}$ sampled at iteration $t^*$, characterized by the highest marginal posterior density function $f(\Lambda, \beta, \Sigma \mid y)$ obtained by integrating out the scale parameters $\tau_0, \gamma_h, \phi_{jh}$ ($j=1,\ldots,p$, $h=1,\ldots$) and the latent factors $\eta_i$ ($i=1,\ldots,n$) from the posterior density function.
Formally, we select the iteration $t^* \in \{1,\ldots,T\}$ such that
\begin{equation*}
f(\Lambda^{(t^*)}, \beta^{(t^*)}, \Sigma^{(t^*)} \mid y)>f(\Lambda^{(t)}, \beta^{(t)}, \Sigma^{(t)} \mid y) \quad (t=1,\ldots,T),
\end{equation*}
where $t=1,\ldots,T$ indexes the posterior samples. Under the structured increasing shrinkage prior described in Section~\ref{subsec:sis-model}, these computations are straightforward. The matrices $\Lambda^{(t^*)}, \beta^{(t^*)}, \Sigma^{(t^*)}$ are Monte Carlo approximations of the maximum \textit{a posteriori} estimator, which corresponds to the Bayes estimator under $L_{\infty}$ loss. 
Although one can argue that $L_{\infty}$ is not an ideal choice of loss philosophically in continuous parameter problems, it nonetheless is an appealing pragmatic choice in our context and is broadly used in other sparse estimation contexts, as in the algorithm proposed by \citet{rockova2016} that similarly aims to recover a strongly sparse posterior mode of an over-parameterized factor model.

\section{Simulation experiments}
\label{sec:sim}
We assess the performance of our structured increasing shrinkage prior compared with current approaches 
\citep{Bhattacharya2011, rockova2016, Legramanti2020}
through a simulation study. 
We have a particular interest in inferring sparse and interpretable loadings matrices $\Lambda$, but also assess performance in estimating the induced covariance matrix $\Omega$ and number of factors.
We generate synthetic data from four scenarios based on different loadings structures. For each scenario we simulate $R=25$ data sets with $n=250$ observations from 
$
y_i \sim N_p(0, \Lambda_0 \Lambda_0^T + I_p) \ (i = 1,\ldots, n)$.
In Scenario a, we assume non sparse $\Lambda_0$, sampling the loadings $\lambda_{jh}$ from a Gaussian distribution with mean zero, variance equal to $\sigma^2_{\lambda}=1$ and ordering them to obtain decreasing variance over the columns. To ensure that each element $\lambda_{jh}$ represents a signal, we shifted them away from zero by $\sigma^2_{\lambda}/3$. In Scenario b we remove the decreasing behaviour and introduce a random sparsity pattern characterized by an increasing number of zero entries over the column index. The loadings matrix for Scenario c is characterized by both the decreasing behaviour over the columns of Scenario a  and the random sparsity structure of Scenario b. Finally, in Scenario d, while the decreasing behaviour is kept, we induce a sparsity pattern dependent on a categorical and two continuous meta covariates $x_0$. Details are reported in Section S2.2 of the Supplementary Material.

For each scenario we consider four combinations of dimension and sparsity level of $\Lambda_0$. We let $(p,k,s) \in \{$(16,4,0.6), (32,8,0.4), (64,12, 0.3), (128,16,0.2)$\}$, where $s$ is the proportion of non-zero entries of $\Lambda$, with the exception of Scenario a where $s=1$. 
\label{rev:comp-time}In these settings the algorithm takes from $0.07$ to $0.73$ seconds of computational time per iteration depending on the dimension $p$ and considering an {\sf R} implementation on an Intel Core i5-6200U CPU laptop with 15.8 GB of RAM.
To estimate the structured increasing shrinkage model, we set $x$ equal to the $p$-variate column vector of $1$s, $\sigma_\beta=1$ and,  consistently with Corollary~\ref{co:asymp-sparse}, $c_p = 2 e \log(p)/p$.
In Scenario d we also estimate and compare a correctly specified structured increasing model with $x=x_0$. For the method proposed by \citet{rockova2016}, we set the hyperparameters as suggested by the authors, while for the remaining approaches, we follow the hyperparameter specification and factor selection guidelines in Section 4 of \citet{schiavon2020}. 

\begin{table}[h]
	\def~{\hphantom{0}}
	\setlength{\tabcolsep}{12pt}
	\caption{Median and interquartile range of LPML and $E(H_a\mid y)$ in 25 replications of Scenario a for different combinations of $(p,k)$; Scenario a is a worst case for the proposed SIS method.}{
		\begin{tabular}{lccccccc}
			& $(p,k)$ & \multicolumn{2}{c}{MGP} & \multicolumn{2}{c}{CUSP} & \multicolumn{2}{c}{SIS} \\
			& & Q$_{0.5}$ & IQR & Q$_{0.5}$ & IQR & Q$_{0.5}$ & IQR \\
			LPML  & ~~(16,4) & ~-28.68 & 0.42 & ~-28.68 & 0.43 & ~-28.65 & 0.41 \\
			& ~~(32,8) & ~-60.08 & 0.45 & ~-60.09 & 0.45 & ~-60.07 & 0.49 \\ 
			& ~(64,12) & -117.68 & 0.56 & -117.75 & 0.53 & -117.88 & 0.56 \\
			& (128,16) & -225.04 & 1.04 & -225.13 & 1.04 & -228.76 & 1.47 \\  
			$E(H_a\mid y)$ & ~~(16,4) & ~~8.17 & 1.44 & ~~4.00 & 0.00 & ~~4.00 & 0.00\\
			& ~~(32,8) & ~10.68 & 0.33 & ~~8.00 & 0.00 & ~~8.00 & 0.00\\    
			& ~(64,12) & ~14.16 & 1.09 & ~12.00 & 0.00 & ~12.00 & 0.00 \\
			& (128,16) & ~17.03 & 0.47 & ~16.00 & 0.00 & ~18.00 & 0.02 
	\end{tabular}}
	\label{tab:scenarioA}       
	{\footnotesize
		LPML, logarithm of the pseudo-marginal likelihood;
		CUSP, cumulative shrinkage process; MGP, multiplicative gamma process; SIS, structured increasing shrinkage process; Q$_{0.5}$, median; IQR, interquartile range.
	}
\end{table}

Scenario a is a worst case for the proposed method since there is no sparsity, no structure, and the elements of the loadings matrix are similar in magnitude.  However, even in this case,  structured increasing shrinkage performs essentially identically to the best competitor, as illustrated by the results in Table~\ref{tab:scenarioA}.  
The results of \citet{rockova2016} are not reported as they are not competitive, as can be seen in table S2 in the Supplementary Material.
We report the median and interquartile range over the $R$ replicates of the logarithm of the pseudo-marginal likelihood  \citep{gelfand1994} and of the estimated posterior mean of the number of factors $E(H_a\mid y)$. 

\begin{figure}[ht]
	\centering
	\includegraphics[width=.8\textwidth]{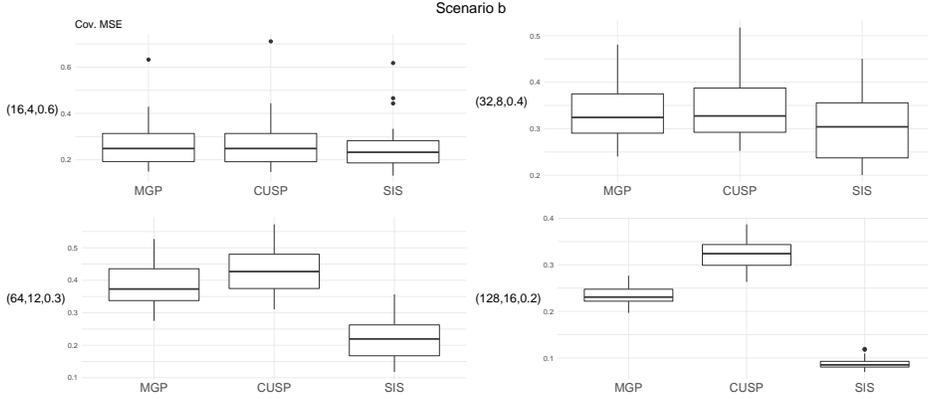} 
	\caption{Boxplots of mean squared error of the covariance matrix of each model for different combinations of $(p,k,s)$ in Scenario b. Cov. MSE, covariance mean squared error; CUSP, cumulative shrinkage process; MGP, multiplicative gamma process; SIS, structured increasing shrinkage process.}
	\label{fig:sim-plot-b}
\end{figure}

Scenario b judges performance in detecting sparsity. The proposed approach shows better performance in the logarithm of the pseudo-marginal likelihood and mean squared error of the covariance matrix, particularly as sparsity increases, as displayed in Fig.~\ref{fig:sim-plot-b}. Consistently with \citep{Legramanti2020}, the covariance mean squared error is estimated in each simulation by $\sum_{j,l}^{p} \sum_{t=1}^S (\omega_{jl}^{(t)}-\omega_{jl0})^2/\{p(p+1)/2\}$, where $\omega_{jl0}$ and $\omega_{jl}^{(t)}$ are the elements $jl$ of $\Omega_0 = \Lambda_0 \Lambda_0^\T +I_p$ and $\Omega^{(t)}= \Lambda^{(t)} \Lambda^{(t)\T} +I_p$, respectively.
The proposed approach allows exact zeros in the loadings, while the competitors require thresholding to infer sparsity.  Following the thresholding approach described in Section S2.2 of the Supplementary Material, we evaluate performance in inferring the sparsity pattern via the mean classification error:
\begin{equation*}
MCE = \frac{1}{S} \sum_{t=1}^{S} \frac{\sum_{j=1}^{p}\sum_{h=1}^{k^{*(t)}} |\mathbbm{1}{(\lambda_{jh0}= 0)} - \mathbbm{1}{(\lambda_{jh}^{(t)}= 0)}|}{p k},
\end{equation*}
where $k^{*(t)}$ is the maximum between the true number of factors $k$ and $H_a^{(t)}$, and $\lambda_{jh0}$ and $\lambda_{jh}^{(t)}$ are the elements $jh$ of $\Lambda_0$ and $\Lambda^{(t)}$, respectively. If $H_a^{(t)}$ or $k$ are smaller than $k^*$, we fix the higher indexed columns at zero, possibly leading to a mean classification error bigger than one.
The results reported in Table~\ref{tab:scenarioB} show that the proposed structured increasing shrinkage prior is much more effective in identifying sparsity in $\Lambda$, maintaining good performance even with large $p$ and in strongly sparse contexts. Also, more accurate estimation of the number of factors is obtained, as reported in Table S1 in the Supplementary Material.

\begin{table}[h]
	\def~{\hphantom{0}}
	\setlength{\tabcolsep}{12pt}
	\caption{Median and interquartile range of the mean classification error computed in 25 replications assuming Scenario b and several combinations of $(p,k,s)$}{
		\begin{tabular}{lccccccc}
			MCE & $(p,k,s)$ & \multicolumn{2}{c}{MGP} &  \multicolumn{2}{c}{CUSP} &  \multicolumn{2}{c}{SIS} \\
			& & Q$_{0.5}$ & IQR & Q$_{0.5}$ & IQR & Q$_{0.5}$ & IQR  \\
			& ~~(16,4,0.6) & 1.06 & 0.16 & 0.38 & 0.01 & 0.24 & 0.09  \\
			& ~~(32,8,0.4) & 0.70 & 0.07 & 0.48 & 0.08 &  0.16 & 0.09 \\ 
			& ~(64,12,0.3) & 0.61 & 0.07 & 0.58 & 0.01 &  0.09 & 0.06 \\
			& (128,16,0.2) & 0.49 & 0.03  & 0.52 & 0.08 & 0.04 & 0.01 
	\end{tabular}}
	\label{tab:scenarioB}       
	{\footnotesize
		MCE, mean classification error; MGP, multiplicative gamma process; CUSP, cumulative shrinkage process; SIS, structured increasing shrinkage process;  Q$_{0.5}$, median; IQR, interquartile range.
	}
\end{table}

Similar comments apply in Scenarios c and d reported in Fig. S2 in the Supplementary Material. The superior performance of the structured increasing shrinkage model is only partially mitigated in Scenario c for large $p$ for the logarithm of the pseudo-marginal likelihood. In Scenario d, the use of meta covariates has a mild benefit in identifying the sparsity pattern. In lower signal-to-noise settings, meta covariates have a bigger impact, and they also aid interpretation, as illustrated in the next section. Additional details, tables, and plots for all scenarios are reported in Section S2.3 of the Supplementary Material.

\section{Finnish bird co-occurrence application}
\label{sec:application}
We illustrate our approach by modelling co-occurrence of the fifty most common bird species in Finland \citep{lindstrom2015},
focusing on data in 2014. 
Response $y$ is an $n\times p$ binary matrix denoting occurrence of $p=50$ species in $n=137$ sampling areas. 
An $n \times c$ environmental covariate matrix $w$ is available, including a 5-level habitat type, `spring temperature' (mean temperature in April and May), and $(\mbox{spring temperature})^2$, leading to $c=7$. We consider a meta covariate $p \times q$ matrix $x$ of species traits: logarithm of typical body mass, migratory strategy (short-distance migrant, resident species, long-distance migrant), and a 7-level superfamily index. We model species presence or absence via a multivariate probit regression model: 
\begin{equation}
\label{eq:app}
y_{ij} = \mathbbm{1}{(z_{ij}>0)}, \quad  z_{ij} = w_i^T \mu_j + \epsilon_{ij},\quad
\epsilon_i = (\epsilon_{i1},\ldots,\epsilon_{ip})^T \sim N_p(0,\Lambda \Lambda^T + I_p), 
\end{equation}
where $\mu_j$ characterizes impact of environmental covariates on species occurrence probabilities, and 
covariance in the latent $z_i$ vector is characterized through a factor model. To borrow information across species while incorporating species traits, we let 
\begin{equation}
\label{eq:app-mean}
\mu_j \sim N_c( b \, x_j , \sigma^2_\mu I_c), \quad b= (b_{1}, \ldots, b_{q}), \quad  b_{m} \sim N_c(0, \sigma^2_b I_c),
\end{equation}
where $b$ is a $c \times q$ coefficient matrix with column vectors $b_m$ given Gaussian priors.

Model \eqref{eq:app}--\eqref{eq:app-mean} is consistent with popular joint species distribution models
\citep{ovaskainen2016, tikhonov2017,ovaskainen2020}, with current standard practice using a multiplicative gamma process for $\Lambda$.  We compare this approach to an analysis that instead uses our proposed structured increasing shrinkage prior to allow the species traits $x$ to impact $\Lambda$ and hence the covariance structure across species.
After standardizing $w$ and $x$, we set $\alpha=4$, $a_\theta=b_\theta=2$ and $\sigma_\mu=\sigma_b=1$. 
Posterior sampling is straightforward via a Gibbs sampler reported in Section S3.1 of the Supplementary Material.

Figure S8 in the Supplementary Material displays the posterior means of $\mu$ and $b$. A first investigation shows large heterogeneity of the habitat type effects across different species. Matrix $b$ shows that covariate effects tend to not depend on
migratory strategy or body mass, with the exception of urban habitats tending to have more
migratory birds.

The estimated $\Lambda$ and meta covariate coefficients $\beta$, following the guidelines of  Section~\ref{subsec:identifiability}, are displayed in Fig.~\ref{fig:application}.
\begin{figure}[h]
	\centering
	\includegraphics[width=.95\textwidth]{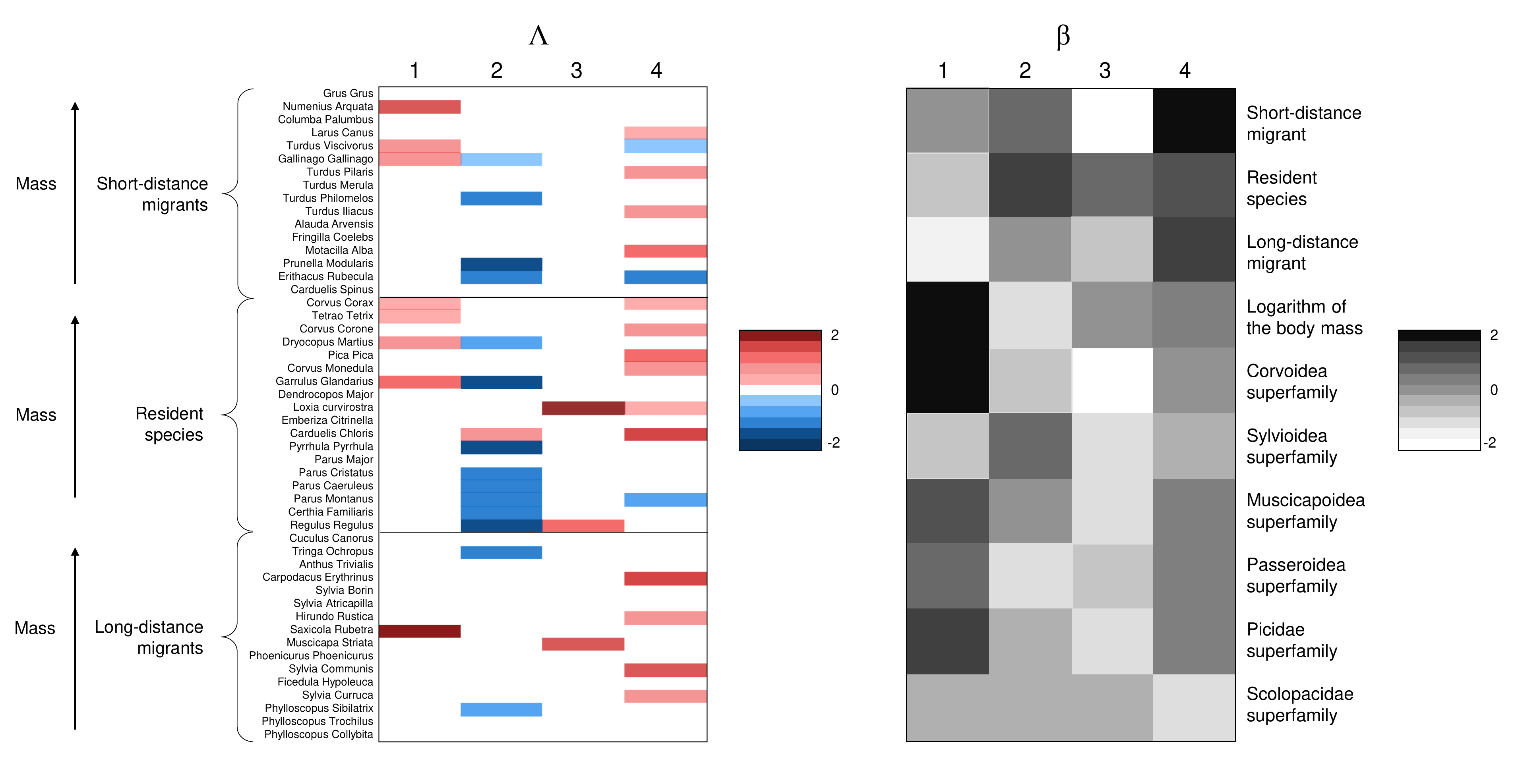}
	\caption{Posterior summaries $\Lambda^{(t^*)}$ and $\beta^{(t^*)}$ of the structured increasing shrinkage model; rows of left matrix refer to 50 birds species, and rows of right matrix to ten species traits. Light coloured cells of $\beta^{(t^*)}$ induce shrinkage on corresponding cells of $\Lambda^{(t^*)}$.}
	\label{fig:application}
\end{figure}
The loadings matrix is quite sparse, indicating that each latent factor impacts a small group of species. Positive sign of the loadings means that high levels of the corresponding factors increase the probability of observing birds from those species.
Lower elements of $\beta^{(t^*)}$, represented with light cells on the right panel, induce higher shrinkage on the corresponding group of birds.
To facilitate interpretation, we re-arrange the rows of $\Lambda^{(t^*)}$ according to the most relevant species traits in terms of shrinkage, which are migration strategy and body mass.
The species influenced by the first factor are fairly homogeneous, characterized by short distance or resident migratory strategies and a larger body mass. 
The strongly negative value of $\beta_{(t^*)42}$ suggests heavier species of birds tend to have loadings close to zero for the second factor.  This is also true for the third factor, which also does not impact short-distance migrants. 

Figure S9 in the Supplementary Material shows a spatial map of the sampling units coloured accordingly to the values of the first and the third latent factors.
We can interpret these latent factors as unobserved environmental covariates. 
We find that the species traits included in our analysis only partially explain the loadings structure; this is as expected and provides motivation for the proposed approach.
\begin{figure}
	\centering
	\includegraphics[width=.95\textwidth]{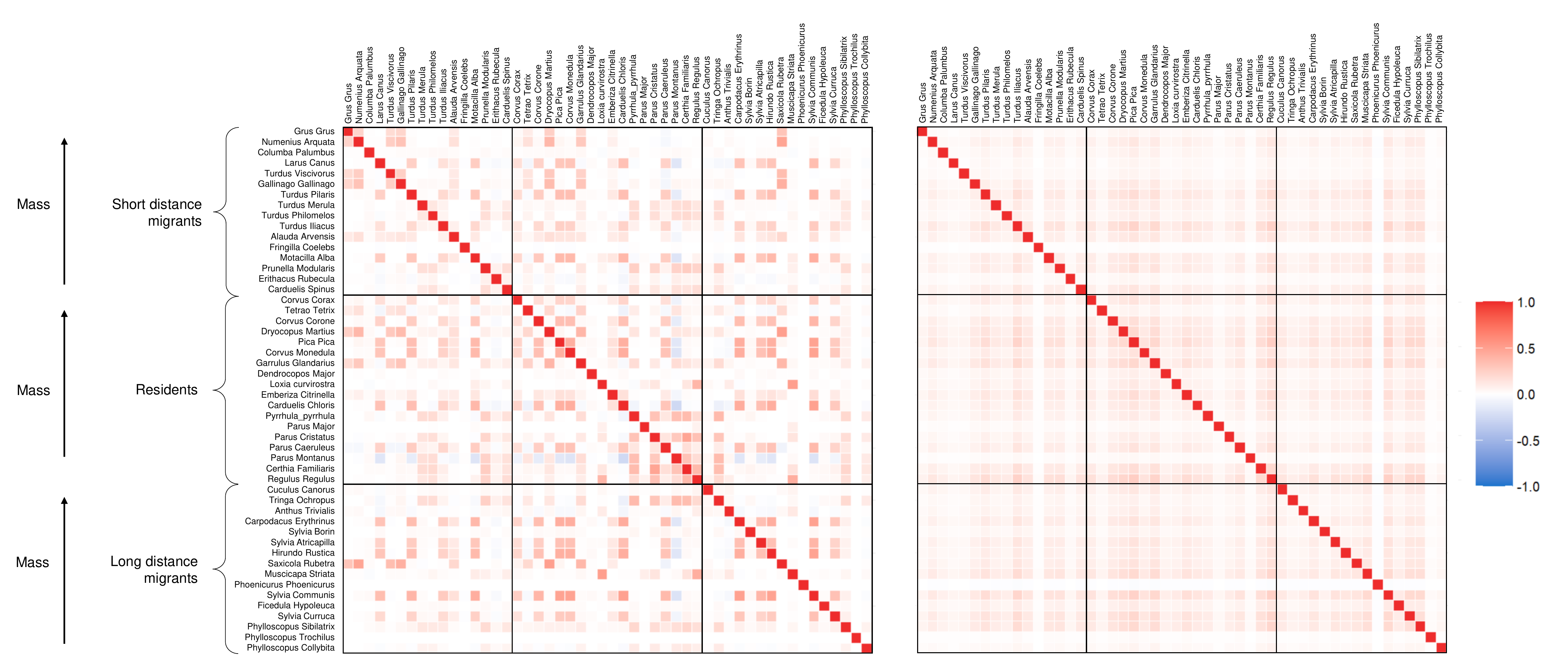} 
	\caption{Posterior mean of the correlation matrices estimated by the structured increasing shrinkage model (on the left) and the multiplicative gamma process model (on the right). 
	}
	\label{fig:application4}
\end{figure}
\begin{figure}
	\centering
	\includegraphics[width=.95\textwidth]{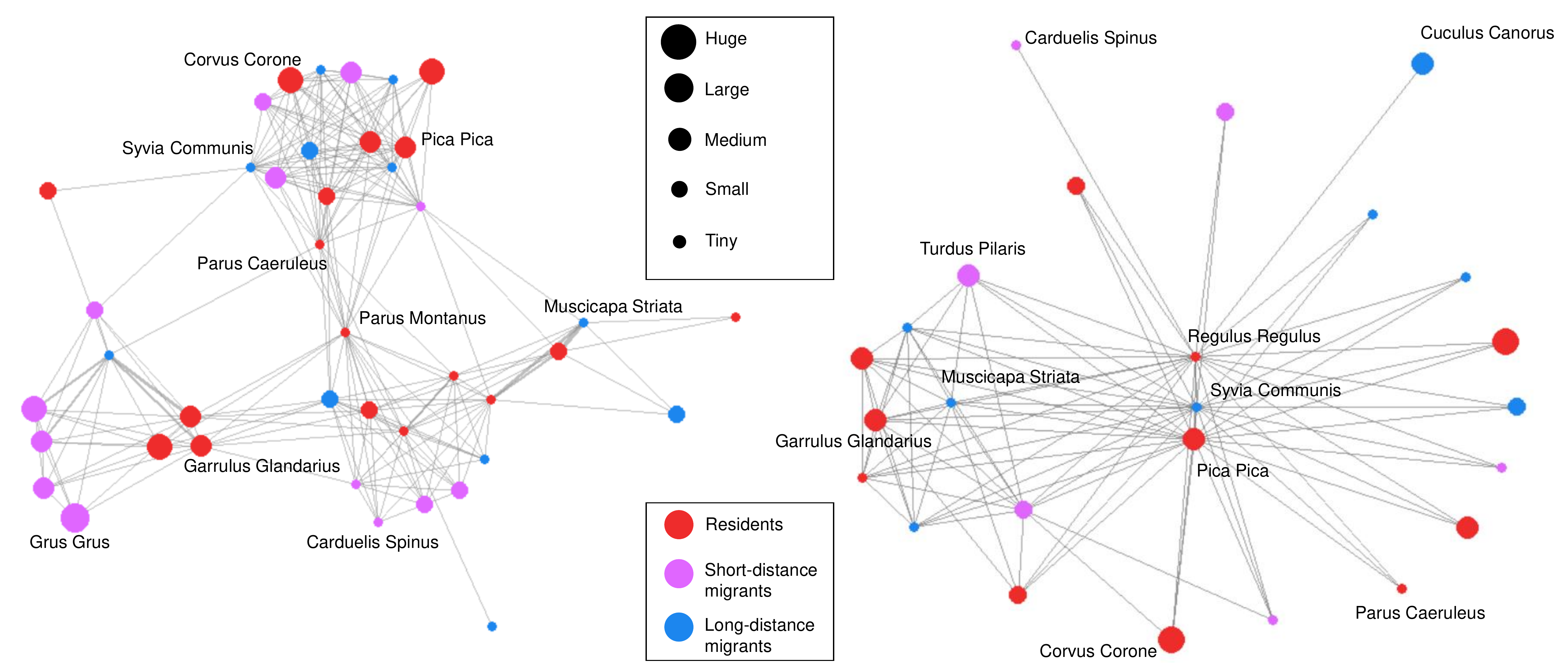}
	\caption{
		Graphical representation based on the inverse of the posterior mean of the correlation matrices estimated by the structured increasing shrinkage model (on the left) and the multiplicative gamma process model (on the right). Edge thicknesses are proportional to the latent partial correlations between species. Values below 0.025 are not reported.  Nodes are positioned using a Fruchterman–Reingold force-direct algorithm. }
	\label{fig:application3}
\end{figure}
Sparsity in the loadings matrix helps in interpretation.  Species may load on the same factor not just because they have similar traits but also because they tend to favor similar habitats for reasons not captured by the measured traits. %


The induced covariance matrix $\Omega = \Lambda \Lambda^T + I_p$ across species is of particular interest.  We compare estimates of $\Omega$ under the multiplicative gamma process, estimated using the {\sf R} package hmsc \citep{tikhonov2020}, and our proposed structured increasing shrinkage model. 
Figure \ref{fig:application4} reports the posterior mean of the correlation matrices under the two competing models.
The network graph based on the posterior mean of the partial correlation matrices, reported in Fig. \ref{fig:application3}, reveals several communities of species under the proposed structured increasing shrinkage prior that are not evident under the multiplicative gamma.  

We also find that the multiplicative gamma process provides a slightly worse fit to the data.  The logarithm of the pseudo marginal likelihood computed on the posterior samples of the structured increasing shrinkage model is equal to $-21.06$, higher than that achieved by the competing model, which is $-21.36$. 
\label{rev:CV-loglik} Using 4-fold cross-validation, we compared the log-likelihood evaluated in the held-out data, with $\mu$ and $\Omega$ estimated by the posterior mean in the training set.
The mean of the log-likelihood was $-22.62$ under the structured increasing shrinkage and $-23.22$ under the multiplicative gamma process prior.

\section*{Acknowledgement}
The authors thank Daniele Durante,  Sirio Legramanti, Otso Ovaskainen, and Gleb Tikhonov for useful comments on an early version of this manuscript. 

\clearpage



\appendix

\section*{Appendix}

\subsection*{Lemmas and proofs}
\begin{proof}[of Theorem~\ref{th:var-dec}] The variance of $\lambda_{jh}$ is
	\begin{equation*}
	\text{var}(\lambda_{j h})=E\{E(\lambda_{j h}^2\mid \phi_{jh}, \gamma_h,  \tau_0)\}= E\{E(\theta_{jh} \mid \phi_{jh}, \gamma_h,  \tau_0)\}.
	\end{equation*}
	By construction, $E(\theta_{jh} \mid \phi_{jh}, \gamma_h,  \tau_0)=\phi_{j h} \gamma_h \tau_0$.
	Then,
	\begin{equation*}
	\text{var}(\lambda_{j h})=E(\phi_{j h} \gamma_h \tau_0)=E(\phi_{j 1}) E(\gamma_h)  E(\tau_0)>E(\phi_{j1}) E(\gamma_{h+1})  E(\tau_0)=\text{var}(\lambda_{j h+1}),
	\end{equation*}
	since the scale parameters are independent and the local scale $\phi_{j h}$ is equally distributed over the column index $h$. 
\end{proof}

To prove Proposition~\ref{pr:power-law} we need to introduce the following Lemma.
\begin{lemma}
	\label{lem:prod}
	Let $u, v$ denote two real positive random variables. If at least one among $(u\mid v)$ and $(v\mid u)$ is power law tail distributed, then the product $uv$ is power law tail distributed.
\end{lemma}

\begin{proof}
	For a positive value $w$, we can write
	\begin{align*}
	\text{pr}(u v > w) &= \int_0^\infty \text{pr}(u >w/v \mid v) f(v) \text{d}v
	= E\{F_{u\mid v}^C (w/v) \},
	\end{align*}
	where $F_{u \mid v}^C(w) = \text{pr}(u>w \mid v)$ and $f(v)$ is the probability density function of $v$. 
	If $F_{u \mid v}^C(w) \geq c \, w^{-\alpha}$ with $c, \alpha$ positive constants and $w$ greater than a sufficiently large number $L$, then
	\begin{equation*}
	\text{pr}(u v > w) \geq E\{ c \, (w/v)^{-\alpha} \}= c \,w^{-\alpha} E(v^{\alpha}) \quad w>L\gg0.
	\end{equation*}
	If $E( v^{\alpha} ) =\infty$, then $\text{pr}(u v > w) > c \,w^{-\alpha}= O(w^{-\alpha})$, otherwise $\text{pr}(u v > w) \geq \nu(w)$ for $w>L$, with $\nu(w)$ a function of order $O(w^{-\alpha})$ as $w$ goes to infinity. This shows that the right tail of the distribution of the random variable $uv$ follows a power law behaviour. 
\end{proof}

\begin{proof}[of Proposition~\ref{pr:power-law}]
	Consider the strictly positive random variables $\theta_{jh}^*=(\theta_{jh} \mid \theta_{jh}>0)$, $\tau_0^*=(\tau_{0} \mid \tau_{0}>0)$, $\gamma_{h}^*=(\gamma_{h} \mid \gamma_{h}>0)$, and $\phi_{jh}^*=(\phi_{jh} \mid \phi_{jh}>0)$.
	Since $\theta_{jh}^{*}$ is equal to the product $\tau^{*} \gamma_h^{*} \phi_{jh}^{*}$ of independent positive random variables,   Lemma~\ref{lem:prod} ensures that if at least one of those scale parameters follows a power law tail distribution, then $\theta_{jh}^*$ is power law tail distributed, so that  $\text{pr}(\theta_{jh}^{*}>\theta)\geq c\, \theta^{-\alpha}$ for $c, \alpha$ positive constants and $\theta>L$.
	Without loss of generality, we focus on the right tail of $\lambda_{jh}$.
	Let 
	\begin{equation}
	\label{eq:lam-dec}
	\text{pr}(\lambda_{j h}>\lambda) = \text{pr}(\lambda_{j h}>\lambda \mid  \theta_{jh}>0) \,\text{pr}(\theta_{jh}>0)+\text{pr}(\lambda_{j h}>\lambda \mid  \theta_{jh}=0) \,\text{pr}(\theta_{jh}=0).
	\end{equation}
	It is straightforward to observe that $\lambda_{j h}$ marginally has a power law tail if and only if
	$(\lambda_{j h} \mid \theta_{jh}>0)$ is power law tail distributed and $\text{pr}(\theta_{jh}>0)$ is strictly positive.
	Since $\text{pr}(\tau_0>0)>0$, $\text{pr}(\gamma_h>0)>0$, and $\text{pr}(\phi_{jh}>0)>0$, then
	$
	\text{pr}(\theta_{jh}>0)>0,
	$
	given independence between the scale parameters.
	Focusing on $\theta_{jh}>0$ in the first term of the right hand side of \eqref{eq:lam-dec}, we have
	$$
	\text{pr}(\lambda_{j h}>\lambda \mid \theta_{jh}^{*})=1- \Phi(\lambda \,\theta_{jh}^{*\,-0.5}),
	$$
	and we want to prove that the marginal $F^{c}_{\lambda_{jh}}(\lambda)=\text{pr}(\lambda_{j h}>\lambda)$ is sub-exponential as $\lambda \rightarrow \infty$. Using the lower bound for the right tail of the standard Gaussian of 
	\cite{abramowitz1948}, 
	\begin{equation*}
	1- \Phi(\lambda \,\theta_{jh}^{*\,-0.5}) \geq \left(\frac{2}{\pi}\right)^{0.5} \frac{\theta_{jh}^{*\,0.5}}{\lambda+(\lambda^2+4\theta_{jh}^*)^{0.5}}  e^{-\lambda^2/(2\theta_{jh}^{*})}.
	\end{equation*}
	Marginalizing over $\theta_{jh}^{*}$, we obtain
	\begin{equation*}
	\text{pr}(\lambda_{j h}>\lambda \mid \theta_{jh}^*) \geq E\left\{\left(\frac{2}{\pi}\right)^{0.5} \frac{\theta_{jh}^{*\,0.5}}{\lambda+(\lambda^2+4\theta_{jh}^{*})^{0.5}}  e^{-\lambda^2/(2\theta_{jh}^{*})}\right\} = E\left\{t_{\lambda}(\theta_{jh}^{*})\right\},
	\end{equation*}
	where $t_{\lambda}(\theta_{jh}^{*})$ is a monotonically increasing nonnegative function defined on the positive real line.
	Applying Markov's inequality, we have $
	E\{t_{\lambda}(\theta_{jh}^{*})\}> \text{pr}(\theta_{jh}^{*}>\epsilon) t_{\lambda}(\epsilon)$,
	and letting $\epsilon=\lambda^2$
	\begin{equation*}
	E\left\{t_\lambda(\theta_{jh}^{*})\right\}> \text{pr}(\theta_{jh}^{*}>\lambda^2) \frac{e^{-0.5}}{1+5^{0.5}} \left(\frac{2}{\pi}\right)^{0.5}.
	\end{equation*}
	If 
	$\text{pr}(\theta_{jh}^{*}>\lambda)\geq  c\, \lambda^{-\alpha}$ for certain $\alpha, c$ positive constants and $\lambda$ sufficiently large, then
	\begin{equation*}
	\text{pr}(\lambda_{j h}>\lambda \mid \theta_{jh}^*) \geq \frac{e^{-0.5}}{1+5^{0.5}} \left(\frac{2}{\pi}\right)^{0.5} c\, \lambda^{-2\alpha} = \tilde{c} \lambda^{-\tilde{\alpha}},
	\end{equation*}
	where $\tilde{c} = e^{-0.5} (1+5^{0.5})^{-1} (2/\pi)^{0.5} c>0$ and $\tilde{\alpha}=\alpha/2>0$. By symmetry, $\text{pr}(\lambda_{j h}< -\lambda \mid \theta_{jh}>0) \geq  \tilde{c} \lambda^{-\tilde{\alpha}}$ for $\lambda>L$ sufficiently large.
	It is sufficient that the marginal distribution of $\theta_{j h}^{*}$ has power law right tail to guarantee that $(\lambda_{j h} \mid \theta_{jh}>0)$ has power law tail and then that marginally $\lambda_{jh}$ has power law tail.
\end{proof}

\begin{proof}[of Theorem~\ref{th:robustness}]
	The mode of the conditional posterior density of $\lambda_{jh}$ is $\tilde{\lambda}_{jh}$ such that
	\begin{equation}
	\label{eq:map}
	l_s(\tilde{\lambda}_{jh}; y, \eta)+ \frac{\partial}{\partial \lambda} \log\{f_{\lambda_{jh} \mid \Lambda_{-jh}} (\lambda)\} \big|_{\lambda=\tilde{\lambda}_{jh}}=0,
	\end{equation}
	where $l_s(\tilde{\lambda}_{jh}; y, \eta)$ is the $jh$th element of the score function of the likelihood for the data $y$ conditionally on the latent variables $\eta$, and $f_{\lambda_{jh} \mid \Lambda_{-jh}}$ is the conditional prior density function of $(\lambda_{jh}\mid \Lambda_{-jh})$.
	Given prior symmetry, without loss of generality, we focus on $\hat{\lambda}_{jh}>0$.
	In a neighbourhood $(\hat{\lambda}_{jh}-\varepsilon, \hat{\lambda}_{jh}+\varepsilon)$ of the conditional maximum likelihood estimate $\hat{\lambda}_{jh}$ of $\lambda_{jh}$, we can approximate the score function using a Taylor expansion:
	$l_s(\lambda; y) = -  \mathcal{J}(\hat{\lambda}_{jh}) \, (\lambda - \hat{\lambda}_{jh}) + \epsilon_\varepsilon$,
	where $\mathcal{J}(\hat{\lambda}_{jh})>0$ is the negative of the derivative of $l_s(\lambda;y)$ evaluated at $\lambda = \hat{\lambda}_{jh} $, and 
	$\epsilon_\varepsilon$ is an approximation error term such that 
	$\lim_{\varepsilon \to 0} \epsilon_\varepsilon/\varepsilon = 0$. 
	For $\hat{\lambda}_{jh}$ large enough, such that $\hat{\lambda}_{jh}-\varepsilon>L$ with $L\gg0$, Lemma \ref{lem:score} holds for every $\lambda$ in $(\hat{\lambda}_{jh}-\varepsilon, \hat{\lambda}_{jh}+\varepsilon)$, leading to the lower bound
	\begin{equation*}
	-\mathcal{J}(\hat{\lambda}_{jh}) \, (\lambda - \hat{\lambda}_{jh}) +  f_{lb}'(\lambda) +\epsilon_{\varepsilon}  \leq l_s(\lambda; y)+ \frac{\partial}{\partial \lambda} \log\{f_{\lambda_{jh} \mid \Lambda_{-jh}} (\lambda)\}, 
	\end{equation*}
	where  $f_{lb}'(\lambda)$ is a non positive continuous function for every $\lambda>0$, $\lim_{\lambda \to +\infty} f_{lb}'(\lambda) = 0$.
	Let $\varepsilon$ be a function of $\hat{\lambda}_{jh}$ such that
	$\lim_{\hat{\lambda}_{jh} \to \infty} \varepsilon = 0$ and
	$\lim_{\hat{\lambda}_{jh}\to \infty} f_{lb}'(\hat{\lambda}_{jh})/ \varepsilon= 0.$   
	The limit for $\hat{\lambda}_{jh} \to \infty$ of the lower bound evaluated in $\hat{\lambda}_{jh}-\varepsilon$ is
	\begin{equation*}
	\lim_{\hat{\lambda}_{jh} \to \infty} \mathcal{J}(\hat{\lambda}_{jh}) \, \varepsilon +  f_{lb}'(\hat{\lambda}_{jh}-\varepsilon) +\epsilon_{\varepsilon} = \lim_{\hat{\lambda}_{jh} \to \infty} |\varepsilon|\, \{\mathcal{J}(\hat{\lambda}_{jh}) +  f_{lb}'(\hat{\lambda}_{jh}-\varepsilon)/|\varepsilon| +\epsilon_{\varepsilon}/|\varepsilon|\}.
	\end{equation*}
	Under Assumption \ref{ass:fisher}, 
	$
	\lim_{\hat{\lambda}_{jh} \to \infty} \mathcal{J}(\hat{\lambda}_{jh}) +  f_{lb}'(\hat{\lambda}_{jh}-\varepsilon)/|\varepsilon| +\epsilon_{\varepsilon}/|\varepsilon| \geq 0,
	$
	which guarantees 
	$
	\hat{\lambda}_{jh}-\varepsilon \leq \tilde{\lambda}_{jh} \leq \hat{\lambda}_{jh},
	$ 
	and, hence  
	$
	\lim_{\hat{\lambda}_{jh} \to \infty}|\tilde{\lambda}_{jh}-\hat{\lambda}_{jh}|=0,
	$
	which proves the theorem.
\end{proof}

\begin{proof}[of Theorem~\ref{th:asym-sparse}]
	Since the local scales are independent, conditionally on $\beta$, we can apply the Chernoff's method and obtain the following upper bound
	\begin{equation*}
	\text{pr}\{|\text{supp}_{\epsilon}(\lambda_{h})|> a s_p \mid \beta_h, \gamma_h, \tau_0\} \leq
	\exp(-s_p a\,t) \,\exp\bigg\{(e^t-1)\, \sum_{j=1}^{p} \zeta_{\epsilon jh}\bigg\},
	\end{equation*}
	for every $t>0$ and $\zeta_{\epsilon jh}=\{\tau_0 \,\gamma_{h}\, g(x_{j}^\T\beta_{h})\}/\epsilon^2$ a function of $\beta_h$.
	Since $g(x_{j}^\T\beta_{h})$ is of order $\le O(\log(p)/p)$ by assumption and is limited above with respect to $\beta_h$, we can deduce
	$g(x_{j}^\T\beta_{h}) \leq  c_{j} \log(p)/p$ for $p$ sufficiently large and for some constant $c_{j}>0$ that does not depend on $\beta_{h}$ and is asymptotically of order $O(1)$ with respect to $p$.
	Then, for $p \gg 0$,
	\begin{align*}
	\sum_{j=1}^{p} g(x_{j}^\T\beta_{h}) &\leq  \sum_{j=1}^{p} c_{j} \log(p)/p
	\leq p\, \log(p)/p \,\max_{1 \leq j \leq p} c_{j}  = m \log(p),
	\end{align*}
	where $m=\max_{1 \leq j \leq p} c_{j}$ does not depend on $\beta_h$.
	Then, the upper bound is
	\begin{equation*}
	\text{pr}\{|\text{supp}_{\epsilon}(\lambda_{h})|> a s_p \mid \beta_h, \gamma_h, \tau_0\} \leq
	\exp\bigg\{-s_p a\,t + (e^t-1)\, \frac{\tau_0 \gamma_h}{\epsilon^2}\, m \log(p) \bigg\}.
	\end{equation*}
	Let us choose $t=\log\{\epsilon^2/(\tau_0 \gamma_h m)+1\}$.
	Since $s_p \geq \log(p) \,c_s$ for a certain $c_s>0$, then, for any $a > (c_s t)^{-1}$, we can write
	\begin{equation*}
	\text{pr}\{|\text{supp}_{\epsilon}(\lambda_{h})|> a s_p \mid \beta_h, \gamma_h, \tau_0\} \leq
	\exp\bigg\{-\log(p) \,\tilde{a} \bigg\},
	\end{equation*}
	where $\tilde{a}$ is a positive constant such that $a = (1+\tilde{a}) (c_s t)^{-1}$.
	The upper bound does not depend on $\beta_h$, so
	\begin{equation*}
	\text{pr}\{|\text{supp}_{\epsilon}(\lambda_{h})|> a s_p \mid \gamma_h, \tau_0\} \leq
	\nu(p)
	\end{equation*}
	with $\nu(p)$ of order $O(p^{-1})$ that goes to zero.
\end{proof}

\nocite{polson2013}

\bibliographystyle{apa}
\bibliography{bib-GIF}

\begin{thebibliography}{}

\bibitem[\protect\astroncite{Abramowitz and Stegun}{1948}]{abramowitz1948}
Abramowitz, M. and Stegun, I.~A. (1948).
\newblock {\em Handbook of mathematical functions with formulas, graphs, and
  mathematical tables}, volume~55.
\newblock US Government printing office.

\bibitem[\protect\astroncite{An et~al.}{2013}]{an2013}
An, X., Yang, Q., and Bentler, P.~M. (2013).
\newblock A latent factor linear mixed model for high-dimensional longitudinal
  data analysis.
\newblock {\em Statistics in medicine}, 32(24):4229--4239.

\bibitem[\protect\astroncite{A{\ss}mann et~al.}{2016}]{assmann2016}
A{\ss}mann, C., Boysen-Hogrefe, J., and Pape, M. (2016).
\newblock Bayesian analysis of static and dynamic factor models: An ex-post
  approach towards the rotation problem.
\newblock {\em Journal of Econometrics}, 192(1):190--206.

\bibitem[\protect\astroncite{Bhattacharya and Dunson}{2011}]{Bhattacharya2011}
Bhattacharya, A. and Dunson, D.~B. (2011).
\newblock {Sparse Bayesian infinite factor models}.
\newblock {\em Biometrika}, 98(2):291--306.

\bibitem[\protect\astroncite{Bhattacharya et~al.}{2015}]{Bhattacharya2015}
Bhattacharya, A., Pati, D., Pillai, N.~S., and Dunson, D.~B. (2015).
\newblock {Dirichlet-Laplace priors for optimal shrinkage}.
\newblock {\em Journal of the American Statistical Association},
  110(512):1479--1490.

\bibitem[\protect\astroncite{Carvalho et~al.}{2010}]{Carvalho2010}
Carvalho, C.~M., Polson, N.~G., and Scott, J.~G. (2010).
\newblock {The horseshoe estimator for sparse signals}.
\newblock {\em Biometrika}, 97(2):465--480.

\bibitem[\protect\astroncite{Dahl}{2006}]{dahl2006}
Dahl, D.~B. (2006).
\newblock Model-based clustering for expression data via a {Dirichlet} process
  mixture model.
\newblock {\em Bayesian inference for gene expression and proteomics},
  4:201--218.

\bibitem[\protect\astroncite{Durante}{2017}]{Durante2017}
Durante, D. (2017).
\newblock {A note on the multiplicative gamma process}.
\newblock {\em Statistics and Probability Letters}, 122:198--204.

\bibitem[\protect\astroncite{Ferrari and Dunson}{2020}]{ferrari2020}
Ferrari, F. and Dunson, D.~B. (2020).
\newblock Bayesian factor analysis for inference on interactions.
\newblock {\em Journal of the American Statistical Association}, pages 1--12.

\bibitem[\protect\astroncite{Gelfand and Dey}{1994}]{gelfand1994}
Gelfand, A.~E. and Dey, D.~K. (1994).
\newblock {Bayesian model choice: asymptotics and exact calculations}.
\newblock {\em Journal of the Royal Statistical Society: Series B
  (Methodological)}, 56(3):501--514.

\bibitem[\protect\astroncite{Jun and Tao}{2013}]{jun2013}
Jun, L. and Tao, D. (2013).
\newblock {Exponential Family Factors for Bayesian Factor Analysis}.
\newblock {\em IEEE Transactions on neural networks and learning systems},
  24(6):964----976.

\bibitem[\protect\astroncite{Legramanti et~al.}{2020}]{Legramanti2020}
Legramanti, S., Durante, D., and Dunson, D.~B. (2020).
\newblock Bayesian cumulative shrinkage for infinite factorizations.
\newblock {\em Biometrika}, 107(3):745--752.

\bibitem[\protect\astroncite{Lindstr{\"o}m et~al.}{2015}]{lindstrom2015}
Lindstr{\"o}m, {\AA}., Green, M., Husby, M., K{\aa}l{\aa}s, J.~A., and
  Lehikoinen, A. (2015).
\newblock Large-scale monitoring of waders on their boreal and arctic breeding
  grounds in northern {Europe}.
\newblock {\em Ardea}, 103(1):3--15.

\bibitem[\protect\astroncite{Liu and Vandenberghe}{2010}]{liu2010}
Liu, Z. and Vandenberghe, L. (2010).
\newblock Interior-point method for nuclear norm approximation with application
  to system identification.
\newblock {\em SIAM Journal on Matrix Analysis and Applications},
  31(3):1235--1256.

\bibitem[\protect\astroncite{Lopes and West}{2004}]{lopes2004}
Lopes, H.~F. and West, M. (2004).
\newblock {Bayesian} model assessment in factor analysis.
\newblock {\em Statistica Sinica}, pages 41--67.

\bibitem[\protect\astroncite{McParland et~al.}{2014}]{mcparland2014}
McParland, D., Gormley, I.~C., McCormick, T.~H., Clark, S.~J., Kabudula, C.~W.,
  and Collinson, M.~A. (2014).
\newblock Clustering south {African} households based on their asset status
  using latent variable models.
\newblock {\em The annals of applied statistics}, 8(2):747.

\bibitem[\protect\astroncite{Miller et~al.}{2019}]{miller2019}
Miller, J.~E., Li, D., LaForgia, M., and Harrison, S. (2019).
\newblock Functional diversity is a passenger but not driver of drought-related
  plant diversity losses in annual grasslands.
\newblock {\em Journal of Ecology}, 107(5):2033--2039.

\bibitem[\protect\astroncite{Miller and Harrison}{2018}]{miller2018}
Miller, J.~W. and Harrison, M.~T. (2018).
\newblock Mixture models with a prior on the number of components.
\newblock {\em Journal of the American Statistical Association},
  113(521):340--356.

\bibitem[\protect\astroncite{Mitchell and Beauchamp}{1988}]{mitchell88}
Mitchell, T.~J. and Beauchamp, J.~J. (1988).
\newblock Bayesian variable selection in linear regression.
\newblock {\em Journal of the American Statistical Association},
  83(404):1023--1036.

\bibitem[\protect\astroncite{Mnih and Salakhutdinov}{2008}]{mnih2008}
Mnih, A. and Salakhutdinov, R.~R. (2008).
\newblock Probabilistic matrix factorization.
\newblock In {\em Advances in neural information processing systems}, pages
  1257--1264.

\bibitem[\protect\astroncite{Montagna et~al.}{2012}]{montagna2012}
Montagna, S., Tokdar, S.~T., Neelon, B., and Dunson, D.~B. (2012).
\newblock {Bayesian} latent factor regression for functional and longitudinal
  data.
\newblock {\em Biometrics}, 68(4):1064--1073.

\bibitem[\protect\astroncite{Murray et~al.}{2013}]{murray2013}
Murray, J.~S., Dunson, D.~B., Carin, L., and Lucas, J.~E. (2013).
\newblock {Bayesian Gaussian copula factor models for mixed data}.
\newblock {\em Journal of the American Statistical Association},
  108(502):656--665.

\bibitem[\protect\astroncite{Ovaskainen and Abrego}{2020}]{ovaskainen2020}
Ovaskainen, O. and Abrego, N. (2020).
\newblock {\em Joint Species Distribution Modelling: With Applications in R}.
\newblock Cambridge University Press.

\bibitem[\protect\astroncite{Ovaskainen et~al.}{2016}]{ovaskainen2016}
Ovaskainen, O., Abrego, N., Halme, P., and Dunson, D. (2016).
\newblock Using latent variable models to identify large networks of
  species-to-species associations at different spatial scales.
\newblock {\em Methods in Ecology and Evolution}, 7(5):549--555.

\bibitem[\protect\astroncite{Polson and Scott}{2010}]{Polson2010}
Polson, N.~G. and Scott, J.~G. (2010).
\newblock {Shrink globally, act locally: Bayesian sparsity and regularization}.
\newblock {\em Bayesian Statistics}, 9:1--16.

\bibitem[\protect\astroncite{Polson et~al.}{2013}]{polson2013}
Polson, N.~G., Scott, J.~G., and Windle, J. (2013).
\newblock {Bayesian} inference for logistic models using p{\'o}lya--gamma
  latent variables.
\newblock {\em Journal of the American statistical Association},
  108(504):1339--1349.

\bibitem[\protect\astroncite{Reich and Bandyopadhyay}{2010}]{reich2010}
Reich, B.~J. and Bandyopadhyay, D. (2010).
\newblock A latent factor model for spatial data with informative missingness.
\newblock {\em The annals of applied statistics}, 4(1):439.

\bibitem[\protect\astroncite{Roberts and Rosenthal}{2007}]{roberts2007}
Roberts, G.~O. and Rosenthal, J.~S. (2007).
\newblock Coupling and ergodicity of adaptive {Markov chain Monte Carlo}
  algorithms.
\newblock {\em Journal of applied probability}, 44(2):458--475.

\bibitem[\protect\astroncite{Ro{\v{c}}kov{\'a} and George}{2016}]{rockova2016}
Ro{\v{c}}kov{\'a}, V. and George, E.~I. (2016).
\newblock Fast bayesian factor analysis via automatic rotations to sparsity.
\newblock {\em Journal of the American Statistical Association},
  111(516):1608--1622.

\bibitem[\protect\astroncite{Rousseau and
  Mengersen}{2011}]{rousseau2011asymptotic}
Rousseau, J. and Mengersen, K. (2011).
\newblock Asymptotic behaviour of the posterior distribution in overfitted
  mixture models.
\newblock {\em Journal of the Royal Statistical Society: Series B (Statistical
  Methodology)}, 73(5):689--710.

\bibitem[\protect\astroncite{Roweis and Ghahramani}{1999}]{roweis1999}
Roweis, S. and Ghahramani, Z. (1999).
\newblock A unifying review of linear {Gaussian} models.
\newblock {\em Neural computation}, 11(2):305--345.

\bibitem[\protect\astroncite{Roy et~al.}{2019}]{roy2019}
Roy, A., Schaich-Borg, J., and Dunson, D.~B. (2019).
\newblock Bayesian time-aligned factor analysis of paired multivariate time
  series.
\newblock {\em arXiv preprint arXiv:1904.12103}.

\bibitem[\protect\astroncite{Schiavon and Canale}{2020}]{schiavon2020}
Schiavon, L. and Canale, A. (2020).
\newblock On the truncation criteria in infinite factor models.
\newblock {\em Stat}, 9:e298.

\bibitem[\protect\astroncite{Thomas et~al.}{2009}]{thomas2009}
Thomas, D.~C., Conti, D.~V., Baurley, J., Nijhout, F., Reed, M., and Ulrich,
  C.~M. (2009).
\newblock Use of pathway information in molecular epidemiology.
\newblock {\em Human genomics}, 4(1):21.

\bibitem[\protect\astroncite{Tikhonov et~al.}{2017}]{tikhonov2017}
Tikhonov, G., Abrego, N., Dunson, D., and Ovaskainen, O. (2017).
\newblock Using joint species distribution models for evaluating how
  species-to-species associations depend on the environmental context.
\newblock {\em Methods in Ecology and Evolution}, 8(4):443--452.

\bibitem[\protect\astroncite{Tikhonov et~al.}{2020}]{tikhonov2020}
Tikhonov, G., Opedal, {\O}.~H., Abrego, N., Lehikoinen, A., de~Jonge, M.~M.,
  Oksanen, J., and Ovaskainen, O. (2020).
\newblock Joint species distribution modelling with the {R-package} hmsc.
\newblock {\em Methods in ecology and evolution}, 11(3):442--447.

\bibitem[\protect\astroncite{Wade et~al.}{2018}]{wade}
Wade, S., Ghahramani, Z., et~al. (2018).
\newblock Bayesian cluster analysis: Point estimation and credible balls (with
  discussion).
\newblock {\em Bayesian Analysis}, 13(2):559--626.

\bibitem[\protect\astroncite{Yang et~al.}{2018}]{yang2018}
Yang, L., Fang, J., Duan, H., Li, H., and Zeng, B. (2018).
\newblock Fast low-rank {Bayesian} matrix completion with hierarchical
  {Gaussian} prior models.
\newblock {\em IEEE Transactions on Signal Processing}, 66(11):2804--2817.

\bibitem[\protect\astroncite{Yuan and Lin}{2006}]{yuan2006}
Yuan, M. and Lin, Y. (2006).
\newblock Model selection and estimation in regression with grouped variables.
\newblock {\em Journal of the Royal Statistical Society: Series B (Statistical
  Methodology)}, 68(1):49--67.

\end{thebibliography}

\clearpage


\title{Supplementary Material for "Generalized infinite factorization models"}
\markboth{}{}
\author{}
\date{}

\maketitle

This supplementary material available at Biometrika online includes the statement and proof of Proposition S1 and the proofs of Proposition~1, Lemmas 1--2, and Corollaries 1--3.
The Gibbs sampling algorithm, settings, and additional results of the simulations and ecology data analysis are reported, including trace plots and a sensitivity analysis to varying hyperparameters.

\renewcommand\thesection{S\arabic{section}}
\renewcommand\thesubsection{S\arabic{section}.\arabic{subsection}} 
\renewcommand{\thefigure}{S\arabic{figure}} 
\renewcommand{\thetable}{S\arabic{table}} 
\renewcommand\theproposition{S\arabic{proposition}} 

\setcounter{section}{0}

\section{Propositions and proofs}
\begin{proposition}
	\label{pr:well-prior}
	Let $\Pi_\Lambda \otimes \Pi_\Sigma$ denote the prior on $(\Lambda, \Sigma)$. Let $\Theta_\Lambda$ and $\Theta_\Sigma$ denote the sample spaces of the matrices $\Lambda$ and $\Sigma$, respectively.
	If $\sum_{h=1}^{\infty} E(\gamma_{h})<\infty$, then, $\Pi_\Lambda \otimes \Pi_\Sigma(\Theta_\Lambda \times \Theta_\Sigma)=1$.
\end{proposition}

\begin{proof}
	Assume $\Sigma \in \Theta_\Sigma$ and  $(\Psi, \Lambda) \in \Theta_\Psi \times \Theta_\Lambda$,  with $\Theta_\Sigma$ the set of $p \times p$ positive semi-definite matrices with finite elements, and 
	\begin{align*}
	\Theta_\Psi \times \Theta_\Lambda &=\bigg\{ \Lambda= (\lambda_{jh}), \, \Psi =(\psi_{hh}) : 
	\sum_{h=1}^{\infty} \lambda_{jh} \psi_{hh} \lambda_{sh}<\infty \; \forall\  j,s \in (1,\ldots,p) \bigg\}.
	\end{align*}
	Due to independence, we can study the prior on $\Sigma$ and $\Lambda$ separately. The prior on $\Sigma$ is defined on the set of positive semi-definite matrices. Therefore, it is sufficient to prove that the elements of $\Lambda \Psi \Lambda^\T$ are finite almost surely. Using Cauchy-Schwartz, it is straightforward to show that all the entries of $\Lambda\Lambda^\T$ are finite if and only if $\sum_{h=1}^{\infty} \psi_{hh} \lambda_{jh}^2 < \infty \ (j=1,\ldots, p)$.
	Let $c$ satisfy 
	$c>\max_{h=1,\ldots,\infty}  \psi_{hh}$ and let $X_j^{\T}$ denote the $j$-th row vector of matrix $X$.
	Since 
	\begin{align*}
	E(\lambda_{j h}^2) &=  E\{E(\lambda_{j h}^2\mid \phi_{jh}, \gamma_h,  \tau_0)\}= E(\phi_{j h}) E(\gamma_h)  E(\tau_0), 
	\\
	E(\phi_{jh})& =  E\{E(\phi_{jh} \mid \beta_h)\}= E\{g(X_j^{\T} \beta_h)\} = E(\phi_{j1}), 
	\end{align*}
	it is sufficient that $\sum_{h=1}^{\infty} E(\gamma_h)<\infty$ to prove that 
	$\sum_{h=1}^{\infty} E(\lambda_{jh}^2) =  E(\phi_{j1}) E(\tau_0) \sum_{h=1}^{\infty} E(\gamma_h) < \infty$ and then $\sum_{h=1}^{\infty} \psi_{hh} \lambda_{jh}^2 < c \sum_{h=1}^{\infty} \lambda_{jh}^2 < \infty$.
\end{proof}

\begin{proof}[of Proposition 1]
	The trace of $\Omega$ is $\text{tr}({\Sigma})+\text{tr}(\Lambda_H \Psi_H \Lambda_H^\T)+\text{tr}(\Lambda_{\Delta_H} \Psi_{\Delta_H} \Lambda_{\Delta_H}^\T)$, where $\Lambda_{\Delta_H} =\Lambda -\Lambda_{H}$ and  $\Psi_{\Delta_H} =\Psi -\Psi_{H}$.
	Hence, it is equivalent to rewrite the probability of interest as
	$$
	\text{pr}\left\{\frac{\text{tr}(\Lambda_{\Delta_H} \Psi_{\Delta_H} \Lambda_{\Delta_H}^\T)}{\text{tr}({\Omega})} \geq 1-T\right\}.
	$$
	By Markov's Inequality
	\begin{equation*}
	\text{pr}\left\{\frac{\text{tr}(\Lambda_{\Delta_H} \Psi_{\Delta_H} \Lambda_{\Delta_H}^\T)}{\text{tr}({\Omega})} \geq 1-T \right\} \leq E\left\{\frac{\text{tr}(\Lambda_{\Delta_H} \Psi_{\Delta_H} \Lambda_{\Delta_H}^\T)}{\text{tr}({\Omega})}\right\}\big/(1-T).
	\end{equation*}
	The expected ratio of two random variables $u$ and $v$ is $E(u/v)= \text{cov}(u, 1/v)+E(u)E(1/v)$,
	which allows us to write $E(u/v)\leq E(u)E(1/v)$ if $\text{cov}(u, 1/v)\leq0$.
	Then, since the covariance between $\text{tr}(\Lambda_{\Delta_H} \Psi_{\Delta_H} \Lambda_{\Delta_H}^\T)$ and $\text{tr}({\Omega})$ is non-negative, the following inequality holds
	\begin{equation*}
	E\left\{\frac{\text{tr}(\Lambda_{\Delta_H} \Psi_{\Delta_H}  \Lambda_{\Delta_H}^\T)}{\text{tr}({\Omega})}\right\} \leq	E\{\text{tr}(\Lambda_{\Delta_H} \Psi_{\Delta_H}  \Lambda_{\Delta_H}^\T)\}
	E\left(\frac{1}{\text{tr}({\Omega})}\right).
	\end{equation*} 
	The trace $\text{tr}(\Lambda_{\Delta_H} \Psi_{\Delta_H}  \Lambda_{\Delta_H}^\T)$ is equal to $\sum_{j=1}^{p} \sum_{h=H+1}^{\infty} \psi_{hh} \lambda_{jh}^2$.
	The variance of $\lambda_{jh}$ is $E(\lambda_{j h}^2)=E(\phi_{j 1}) E(\gamma_h)  E(\tau_0)$.   	Let $c$ satisfy $c \geq \max_{h=1,\ldots,\infty}  \psi_{hh}$.
	Since $E(\phi_{j 1})$ is finite and $E(\gamma_{h})=a b^{h-1}$ with $a,b$ positive constants and $b<1$, then
	\begin{equation*}
	E\{\text{tr}(\Lambda_{\Delta_H} \Psi_{\Delta_H} \Lambda_{\Delta_H}^\T)\} \leq c E(\tau_0) a \frac{b^H}{1-b} \sum_{j=1}^{p}  E(\phi_{j 1}).
	\end{equation*} 
	Since $\text{tr}({\Omega})=\text{tr}(\Lambda \Psi \Lambda^\T)+\text{tr}(\Sigma)$, we 
	know that $\text{tr}({\Omega}) \geq   \sum_{h=1}^{\infty} \psi_{hh} \lambda_{jh}^2 + \sigma_j^2$ for any $j$ in $1,\ldots,p$, where $\sigma_j^2$ is the $j$-th diagonal element of $\Sigma$. Then, for any $j$ in $1,\ldots,p$,  we obtain
	\begin{equation*}
	\frac{1}{\text{tr}({\Omega})} \leq \frac{1}{ \sum_{h=1}^{\infty} \psi_{hh} \lambda_{jh}^2 + \sigma_j^2},
	\end{equation*}
	and, consequently,
	\begin{equation*}
	E\left\{\frac{1}{\text{tr}({\Omega})}\right\} \leq E\left(\sigma_j^{-2}\right), \qquad E\left\{\frac{1}{\text{tr}({\Omega})}\right\} \leq E\left\{\left(\sum_{h=1}^{\infty} \psi_{hh} \lambda_{jh}^2\right)^{-1}\right\}.
	\end{equation*} 
	Therefore, since $m_{\Omega} = \min_{j=1,\ldots,p} \left[E(\sigma_j^{-2}), E\left\{\left(\sum_{h=1}^{\infty} \psi_{hh} \lambda_{jh}^2\right)^{-1}\right\}\right]< \infty$, then
	\begin{equation*}
	\text{pr}\left\{\frac{\text{tr}(\Lambda_H \Psi_H \Lambda_H^\T)+\text{tr}(\Sigma)}{\text{tr}({\Omega})}\leq T\right\}\leq \bigg(\frac{1}{1-T}\bigg)  a \,c \,\frac{b^H}{1-b} \, m_{\Omega}\,   E(\tau_0) \, \sum_{j=1}^{p}  E(\phi_{j 1}),
	\end{equation*} 
	as stated by the Theorem.
\end{proof}

\begin{proof}[of Lemma~1]
	Consistently with Proposition 2,  $(\lambda_{jh} \mid \Lambda_{-jh})$ has power law tail if $(\theta_{jh}\mid  \Lambda_{-jh})$ has power law tail. 
	Furthermore, $\text{pr}(|\lambda_{j h}|>\lambda \mid \Lambda_{-jh})$ has power law tail for large $\lambda$ if and only if
	$\text{pr}(|\lambda_{j h}|>\lambda \mid \Lambda_{-jh}, \theta_{jh}>0)$ has power law tail and $\text{pr}(\theta_{jh}>0 \mid \Lambda_{-jh})>0$.
	The latter inequality is always true when the marginal $\text{pr}(\theta_{jh}>0)$ is positive.
	To prove $(\theta_{jh}\mid  \Lambda_{-jh})$ has power law tail, we apply Lemma 3.
	We first focus on proving the lemma when $\phi_{jh}$ satisfies the power law tail condition with $\tau_0 \gamma_h = w_{h}.$ As the local scale $\phi_{jh}$ is independent from $(\Lambda_{-jh}, w_{h})$ given $\beta_h$, its conditional density is
	\begin{equation*}
	f_{\phi_{jh}\mid w_{h}, \Lambda_{-(jh)}}(\phi) =  \int_\Re f_{\phi_{jh} \mid \beta_h}(\phi)\, f_{\beta_h \mid w_{h}, \Lambda_{-jh}}(\beta) \,\text{d}\beta.
	\end{equation*}
	As the tail conditions hold for any possible prior on $\beta$, we have 
	\begin{equation*}
	f_{\phi_{jh}}(\phi) = \int_\Re f_{\phi_{jh} \mid \beta_h}(x) \, f (\beta)\, \text{d}\beta, \qquad
	f_{\phi_{jh}}(\tilde{\phi}) \propto \tilde{\phi}^{-\alpha}, \qquad \tilde{\phi}=\{\phi: \phi>l\}, \qquad  L\gg0,
	\end{equation*}
	for any prior density $f$ defined on $\Re$.
	Hence, $(\phi_{jh}\mid w_{jh}, \Lambda_{-jh})$ is power law tail distributed.
	We now focus on proving the lemma when $\tau_0$ or $\gamma_h$ are power law tail distributed. 
	Let $r_{h}^*=(r_{h}\mid r_{h}>0)$ and $w_{jh}^*=(w_{jh}\mid w_{jh}>0)$, where $r_h$ is the scale parameter with power law tail and $w_{h}$ is the product of the remaining two scale parameters, respectively.
	By Bayes' Theorem 
	\begin{equation*}
	f_{r_{h}^*\mid w_{jh}^*, \Lambda_{-jh}}(r)= \frac{f_{ \Lambda_{-jh}\mid w_{jh}^*, r_{h}^*}(\Lambda_{-jh}; r) f_{ r_{h}^* \mid w_{jh}^*}(r) }{f_{\Lambda_{-jh}\mid w_{jh}^*}(\Lambda_{-jh})}.
	\end{equation*}
	Since $r_{h}^*$ is independent from $w_{jh}^*$ for any parameter scale, it is sufficient to prove that the function $f_{ \Lambda_{-jh}\mid w_{jh}^*, r_{h}^*}(\Lambda_{-jh};r)$ decreases slower than $c\, r^{\,-\alpha}$, for $c,\alpha>$ positive constants, when $r \rightarrow \infty$.
	Denoting $F_{\tau_0, \phi_{11} \ldots \phi_{p k}, \gamma_{1},\ldots,\gamma_{h-1},\gamma_{h+1},\ldots,\gamma_{k} \mid w_{jh}^*, r_{h}^{*}}$ the probability measure for conditional density $f_{\tau_0, \phi_{11} \ldots \phi_{p k}, \gamma_{1},\ldots,\gamma_{h-1},\gamma_{h+1},\ldots,\gamma_{k} \mid w_{jh}^*, r_{h}^{*}}$, we can write
	\begin{align*}
	f_{ \Lambda_{-jh}\mid w_{jh}^{*}, r_{h}^*}(\Lambda_{-jh};r) &= \int f_{ \Lambda_{-jh}\mid \tau_0, \phi_{11} \ldots \phi_{p k}, \gamma_{1},\ldots,\gamma_{k}} (\Lambda_{-jh};r) \, \text{d} F_{\tau_0, \phi_{11} \ldots \phi_{p k}, \gamma_{1},\ldots,\gamma_{h-1},\gamma_{h+1},\ldots,\gamma_{k} \mid w_{jh}^*, r_{h}^{*}}\\
	&= \int \prod_{(s,m) \neq (j,h)}  f_{\lambda_{sm} \mid \theta_{sm}}(\lambda_{sm};r) \, \text{d} F_{\tau_0, \phi_{11},\ldots,\gamma_{k} \mid w_{jh}^*, r_{h}^{*}}\\
	&= E \bigg\{ \prod_{(s,m) \neq (j,h)} f_{\lambda_{sm} \mid \theta_{sm}}(\lambda_{sm};r) \biggm| w_{jh}^*, r_{h}^{*}, \Lambda_{-jh} \bigg\}
	\end{align*}
	The product inside the expectation is zero when there is a pair of indices $(s,m)$ such that $\lambda_{sm}\neq 0$ and $\theta_{sm} =0$. However, since the probability $\text{pr}(\theta_{sm}=0 \mid \lambda_{sm}\neq 0)=0$, we know that 
	the expected value of the product between the functions $f_{\lambda_{sm} \mid \theta_{sm}}(\lambda_{sm};r)$, given $w_{jh}^*, r_{h}^{*}, \Lambda_{-jh}$, is strictly positive.
	We first focus on the case $\gamma_h=r_h$ and prove that $f_{ \Lambda_{-jh}\mid w_{jh}^*, \gamma_{h}^*}(\Lambda_{-jh};\gamma)$ decreases slower than $c \gamma^{-\alpha}$ for $c,\alpha>0$.
	In this case, we can write the above expectation as
	\begin{equation*}
	E \bigg\{\prod_{s=1, m \neq h}^p f_k(\lambda_{sm})\, \prod_{s\neq j} f_{\lambda_{sh} \mid \theta_{sh}}(\lambda_{sh};\gamma_h^*) \biggm| w_{jh}^*, \gamma_{h}^{*}, \Lambda_{-jh} \bigg\},
	\end{equation*}
	where $\prod_{s=1,m \neq h}^p f_k(\lambda_{sm})$ is a product between $(k-1)\times p$ strictly positive random variables that does not depend on $w_{jh}^*$ and $\gamma_h^*$, while $\prod_{s\neq j} f_{\lambda_{sh} \mid \theta_{sh}}(\lambda_{sh};\gamma_h^*)$ is a product between $p$ strictly positive random variables. In particular, if $w_{sh}=0$, then $f_{\lambda_{sm} \mid \theta_{sh}}(\lambda_{sh};\gamma_h^*)=\mathbbm{1}{(\lambda_{sh}=0)}$.
	If $w_{sh}>0$, then
	\begin{equation*}
	f_{\lambda_{sh} \mid \theta_{sh}}(\lambda_{sh};\gamma_h^*) = (2\pi w_{sh}^{*} \gamma_{h}^*)^{-0.5} \exp\bigg(-\frac{\lambda_{sh}^2}{2w_{sh}^*\gamma_{h}^{*}}\bigg) > (2\pi w_{sh}^{*} \gamma_{h}^*)^{-0.5} \exp\bigg(-\frac{\lambda_{sh}^2}{2w_{sh}^*}\bigg).
	\end{equation*}
	Therefore, the upper bound
	\begin{equation*}
	f_{\lambda_{sh} \mid \theta_{sh}}(\lambda_{sh};\gamma_h^*) \geq \begin{cases}
	\min \{1, \, (2\pi w_{sh}^{*} \gamma_{h}^*)^{-0.5} \}, \qquad &\text{if} \;\lambda_{sh = 0}\\
	(2\pi w_{sh}^{*} \gamma_{h}^*)^{-0.5} \exp\{-\lambda_{sh}^2/ (2w_{sh}^*)\} \qquad &\text{if} \;\lambda_{sh \neq 0},
	\end{cases}
	\end{equation*}
	holds with probability equal to 1. For $\gamma>1$, we note that $f_{\lambda_{sh} \mid \theta_{sh}}(\lambda_{sh};\gamma) \geq \gamma^{-0.5} u_{\lambda_{sh}}$ with
	\begin{equation*}
	u_{\lambda_{sh}} = \begin{cases}
	\min \{1, \, (2\pi w_{sh}^{*})^{-0.5} \}, \qquad  &\text{if} \;\lambda_{sh}= 0,\\
	(2\pi w_{sh}^{*})^{-0.5} \exp\{-\lambda_{sh}^2/ (2w_{sh}^*)\} \qquad &\text{if} \;\lambda_{sh} \neq 0.
	\end{cases}
	\end{equation*}
	Then,  
	\begin{align*}
	E \bigg\{ \prod_{(s,m) \neq (j,h)} f_{\lambda_{sm} \mid \theta_{sm}}(\lambda_{sm};\gamma_h^{*}) &\biggm| w_{jh}^*, \gamma_{h}^{*}, \Lambda_{-jh} \bigg\} \geq \\
	&E \bigg\{ \prod_{s=1, m \neq h}^p f_k(\lambda_{sm})\, \prod_{s\neq j} \gamma_h^{*-0.5} u_{\lambda_{sh}} \biggm| w_{jh}^*, \gamma_{h}^{*}, \Lambda_{-jh} \bigg\}=\\
	&\gamma_h^{*-0.5\,(p-1)} E \bigg\{ \prod_{s=1, m \neq h}^p f_k(\lambda_{sm})\, \prod_{s\neq j} u_{\lambda_{sh}} \biggm| w_{jh}^*, \Lambda_{-jh} \bigg\},
	\end{align*}
	where the expectation is strictly positive and not depending on $\gamma_h$. Therefore, for $\gamma$ sufficiently large,   $f_{\Lambda_{-jh}\mid w_{jh}^*, \gamma_h^{*}}(\Lambda_{-jh};\gamma) \geq c \gamma^{-\alpha}, \qquad c,\alpha >0$
	holds, so that $(\gamma_{h}\mid w_{jh}, \Lambda_{-jh})$ is power law tail distributed.
	Similarly, if $\tau_0=r_h$ ($h=1,\ldots,H$), 
	\begin{align*}
	E \bigg\{ \prod_{(s,m) \neq (j,h)} f_{\lambda_{sm} \mid \theta_{sm}}(\lambda_{sm};\tau_0^*) &\mid w_{jh}^*, \tau_{0}^{*}, \Lambda_{-jh} \bigg\} \geq\\
	&E \bigg( \prod_{(s,m) \neq (j,h)} \tau_0^{*-0.5} u_{\lambda_{sm}} \mid w_{jh}^*, \tau_0^{*}, \Lambda_{-jh} \bigg)=\\
	&\tau_0^{*-0.5\,(pH-1)} E \bigg( \prod_{(s,m) \neq (j,h)} u_{\lambda_{sm}} \mid w_{jh}^*, \Lambda_{-jh} \bigg),
	\end{align*}
	where
	\begin{equation*}
	u_{\lambda_{sm}} = \begin{cases}
	\min \{1, \, (2\pi w_{sm}^{*})^{-0.5} \}, \qquad \qquad   &\text{if} \;\lambda_{sm}= 0,\\
	(2\pi w_{sm}^{*})^{-0.5} \exp\{-\lambda_{sm}^2/ (2w_{sm}^*)\}, \qquad &\text{if} \;\lambda_{sm} \neq 0.
	\end{cases}
	\end{equation*}
	is strictly positive and does not depend on $\tau_0$.
	Then, if the number $H$ of columns of $\Lambda_H$ is finite, $f_{\Lambda_{-jh}\mid w_{jh}^*, \tau_0^{*}}(\Lambda_{-jh};\tau) \geq c \tau^{-\alpha}$ with $c,\alpha >0$ and $\tau$ sufficiently large, implying $(\tau_{0}\mid w_{jh}, \Lambda_{-jh})$ is power law tail distributed.
	Hence, if any of the scale parameters is power law tail distributed for any prior on $\beta$, then its distribution, conditionally on $\Lambda_{-jh}$ and on the product of the other two parameters, is power law tail distributed and, as a consequence, $(\lambda_{jh} \mid \Lambda_{-jh})$ is power law tail distributed.
	Since $f_{ \lambda_{jh} \mid \Lambda_{-jh}}(\lambda) \geq c|\lambda|^{-\alpha}$ for certain $c,\alpha$ positive constants and $|\lambda|>L$ sufficiently large, in the same settings, we can write
	$$
	f_{ \lambda_{jh} \mid \Lambda_{-(jh)}}(\lambda) = c |\lambda|^{-\alpha}\{1+t(|\lambda|)\},
	$$
	where $t(|\lambda|)$ is a positive function.
	Then,
	\begin{alignat*}{3}
	\frac{\partial [\log \{f_{\lambda_{jh}\mid \Lambda_{-(jh)} }(\lambda)\}]}{\partial \lambda} &= -\frac{\alpha}{\lambda} + \frac{\partial t(\lambda) }{\partial \lambda} \qquad
	&&\text{for} \; \lambda >L \quad &&\text{and} \; L\gg0 ,\\
	\frac{\partial [\log \{f_{\lambda_{jh}\mid \Lambda_{-(jh)} }(\lambda)\}]}{\partial \lambda} &= \frac{\alpha}{\lambda} + \frac{\partial \{-t(\lambda)\} }{\partial \lambda} \qquad
	&&\text{for} \; \lambda < -L \quad &&\text{and} \; L\gg0 ,
	\end{alignat*}
	We now consider the sign of the derivative of $t(|\lambda|)$. If $t(|\lambda|)$ is not decreasing, 
	\begin{alignat*}{3}
	\frac{\partial [\log \{f_{\lambda_{jh}\mid \Lambda_{-(jh)} }(\lambda)\}]}{\partial \lambda} &\geq -\frac{\alpha}{\lambda}, \qquad
	&&\text{for} \; \lambda >L \quad &&\text{and} \; L\gg0 ,\\
	\frac{\partial [\log \{f_{\lambda_{jh}\mid \Lambda_{-(jh)} }(\lambda)\}]}{\partial \lambda} &\leq \frac{\alpha}{\lambda}, \qquad
	&&\text{for} \; \lambda < -L \quad &&\text{and} \; L\gg0 ,
	\end{alignat*}
	whereas if $t(|\lambda|)$ is decreasing, its derivative goes to zero when $|\lambda|$ goes to infinity. Therefore,
	\begin{alignat*}{3}
	\frac{\partial [\log \{f_{\lambda_{jh}\mid \Lambda_{-(jh)} }(\lambda)\}]}{\partial \lambda} &\geq f_{lb}'(\lambda) \qquad
	&&\text{for} \; \lambda >L \quad &&\text{and} \; L\gg0 ,\\
	\frac{\partial [\log \{f_{\lambda_{jh}\mid \Lambda_{-(jh)} }(\lambda)\}]}{\partial \lambda} &\leq -f_{lb}'(|\lambda|) \qquad
	&&\text{for} \; \lambda <-L \quad &&\text{and} \; L\gg0 ,
	\end{alignat*}
	where $f_{lb}'(\lambda)<0 \; \forall\, \lambda>0$ and $\lim_{\lambda \to \infty} f_{lb}'(|\lambda|) = 0$.
	The proof is concluded by using this result along with the fact that $f_{ \lambda_{jh} \mid \Lambda_{-(jh)}}(|\lambda|)$ is decreasing when $\lambda \rightarrow \infty$,  
	\begin{alignat*}{3}
	\frac{\partial [ \log \{f_{ \lambda_{jh} \mid \Lambda_{-(jh)}}(\lambda)\}]}{\partial \lambda} &\leq 0 \qquad
	&&\text{for} \; \lambda >L \quad &&\text{and} \; L\gg0 \\
	\frac{\partial [\log \{f_{ \lambda_{jh} \mid \Lambda_{-(jh)}}(\lambda)\}]}{\partial \lambda} &\geq 0 \qquad
	&&\text{for} \; \lambda >-L \quad &&\text{and} \; L\gg0, 
	\end{alignat*}
	showing that the limit of the derivative for $|\lambda| \to \infty$ is equal to zero.
\end{proof}

\begin{proof}[of Lemma 2]
	In both the multiplicative gamma process and cumulative shrinkage process, priors on $\Lambda$ are exchangeable within columns, that is $\text{pr}(|\lambda_{j h}|>\epsilon \mid \gamma_{h}, \tau_0)= \zeta_{\epsilon h}$ does not depend on $j$. 
	Then, the prior density of $|\text{supp}_{\epsilon}(\lambda_{h})|$, conditionally on $\gamma_h$ and $\tau_0$ is {\em a priori} distributed as a sum of independent and identically distributed Bernoulli random variables $\text{Ber}(\zeta_{\epsilon h})$. Furthermore, $\zeta_{\epsilon h}$ does not depend on $p$. By applying the Chernoff's method, we obtain
	\begin{equation*}
	\text{pr}\{|\text{supp}_{\epsilon}(\lambda_{h})|< a s_p \mid \gamma_h, \tau_0\} \leq \exp\big\{a t s_p +p \zeta_{\epsilon h} (e^{-t}-1)\big\},
	\end{equation*}
	for any $t>0$ and with $1-e^{-t}>0$. 
	Hence,
	\begin{equation*}
	\text{pr}\{|\text{supp}_{\epsilon}(\lambda_{h})|> a s_p \mid \gamma_h, \tau_0\} \geq 1-\exp[-p \{(1-e^{-t}) \zeta_{\epsilon h}-a t s_p/p \}],
	\end{equation*}
	where the limit of the lower bound is
	$\lim_{p \to \infty} 1-\exp[-p \{(1-e^{-t}) \zeta_{\epsilon h}-a t s_p/p \}]= 1$,
	which concludes the proof.
\end{proof}

\begin{proof}[of Corollary 1]
	
	i. It is sufficient to prove the conditions required by Theorem 1. We have
	$
	E(\gamma_h)=E(\vartheta_h)E(\rho_h)=E(\rho_h)\, b_\theta/(a_\theta -1),
	$
	where
	\begin{equation*}
	E(\rho_h)=1-\sum_{l=1}^{h}E(w_l)=1-\sum_{l=1}^{h-1}E(w_l)-E(w_h)=E(\rho_{h-1})-E(w_h).
	\end{equation*}
	Since the random variable $w_l$ is obtained as a product of positive random variables, $E(w_l)>0$ for every $l=1,\ldots,h$.
	Therefore $E(\gamma_{h}) < E(\gamma_{h-1})$ for each $h=2,\ldots,H$.
	
	ii. It is sufficient to prove the conditions required by Proposition 1. It is straightforward to verify $E(\tau_0)=1$ and $E(\phi_{jh})\leq 1$ for $j=1,\ldots,p$ and $h=1,\ldots,\infty$. The column scale expectation is
	\begin{equation*}
	E(\gamma_h)=E(\vartheta_h) \left(\frac{\alpha}{1+\alpha}\right) \left(\frac{\alpha}{1+\alpha}\right)^{h-1} ,
	\end{equation*}
	which can be written in a form $ab^{h-1}$.
	The elements $\sigma_j^{-2}$ are gamma distributed guaranteeing finite expectation for all $j=1,\ldots,p$.
\end{proof}

\begin{proof}[of Corollary 2]
	It is sufficient to prove the conditions required by Theorem 2.
	The probability density function of the column scale $\gamma_h$ ($h=1,2,\ldots$) of model \eqref{eq:CUSPadj} evaluated at a certain $\gamma>0$ is 
	\begin{equation*}
	f_{\gamma_h}(\gamma)= \text{pr}(\rho_h = 1) f_{\vartheta_h}(\gamma) \propto  \gamma^{-a_\theta-1} \exp(-b_\theta/\gamma),
	\end{equation*}
	where $f_{\vartheta_h}(\gamma)$ is the inverse gamma probability density function evaluated at $\gamma$.
	The function $\gamma^{-a_\theta-1} \exp(-b_\theta/\gamma)$ is of order $O\{\gamma^{-(a_\theta+1)}\}$ as $\gamma$ goes to infinity. Since $a_\theta>0$, we conclude that the column scale $\gamma_h$ is power law tail distributed. The independence between $\gamma_h$ and $\beta_h$ ($h=1,2,\ldots$) guarantees that the latter result hold for any possible prior distribution $f_\beta$ on $\beta$.
\end{proof}

\begin{proof}[of Corollary 3]
	It is sufficient to prove the conditions required by Theorem 3.
	The structured increasing shrinkage prior is such that, for every $j=1,\ldots,p$ and $h\geq 1$, we have $g(x_j^T\beta_h) \leq c_p < 1$.
	The proof is obtained under the assumption $c_p = O\{\log(p)/p\}$. 
\end{proof}

\section{Simulation experiments}

\subsection{Gibbs sampler for structured increasing shrinkage model for Gaussian data}

We can rewrite the model for $y_{ij}$ for the specific case of the structured increasing shrinkage process and Gaussian data as 
\begin{equation*}
y_{ij}= \sum_{h=1}^\infty \sqrt{\rho_{h}} \sqrt{\phi_{jh}} \,  \lambda_{jh}^* \eta_{ih} + \epsilon_{ij} \qquad \lambda_{jh}^* \sim N(0, \vartheta_h),
\end{equation*}
where $\lambda_{jh}^*$ is a continuous random variable and we let 
$\beta_{mh} \sim N(0,\sigma_\beta^2)$.
The notation $(x\mid-)$ denotes the full conditional distribution of $x$ conditionally on everything else.
Given $H$ the number of factors of the truncated model, the sampler cycles through the following steps.
\Step{1}{
	Update, for $i=1, \ldots, n$, the factor $\eta_i$ according to the posterior full conditional 
	\begin{equation*}
	(\eta_{i}\mid-) \sim N_{H}\big\{(I_H+\Lambda_H^\T \Sigma^{-1} \Lambda_H)^{-1} \Lambda_H^\T \Sigma^{-1} y_{i},\, (I_H+\Lambda_H^\T \Sigma^{-1} \Lambda_H)^{-1}\big\}.
	\end{equation*} 
}
\Step{2}{
	Update, for $j$ in $1, \ldots,p$, the elements of $\Sigma$, by sampling 
	\begin{equation*}
	(\sigma_{j}^{-2}\mid-) \sim \text{Ga}\left\{a_\sigma +\frac{n}{2},\, b_\sigma +\frac{1}{2} \sum_{i=1}^{n} (y_{ij}-\lambda_{j}^\T \eta_i)^2 \right \}.
	\end{equation*} 
}
\Step{3}{
	Update $\beta_h$ ($h=1,\ldots,H$) exploiting  the P\'{o}lya-Gamma data-augmentation strategy \citep{polson2013} and the decompostition $\phi_{jh} =\phi_{jh}^{(L)} \phi_{jh}^{(C)}$, with $\phi_{jh}^{(L)} \phi_{jh}^{(C)}$ independent a priori and distributed as Ber$\{\text{logit}^{-1}(x_j^\T \beta_h)\}$ and Ber$(c_p)$, respectively.
}
\Substep{3}{1}{
	Update $\phi_{jh}^{(L)}$, for $j=1,\ldots,p$ and $h=1,\ldots,H$, setting $\phi_{jh}^{(L)}=1$ if $\phi_{jh}=1$ and sampling from the full conditional distribution
	\begin{equation*}
	\text{pr}(\phi_{jh}^{(L)}=l) \propto 
	\begin{cases}
	1-\text{logit}^{-1}(x_j^\T \beta_h)\qquad \qquad \qquad \quad \, \text{for} \;
	l=0, \\
	\text{logit}^{-1}(x_j^\T \beta_h)(1-c_p)  \qquad \,   \text{for} \;l=1,\\
	\end{cases}
	\end{equation*}		
	if $\phi_{jh}=0$.
}
\Substep{3}{2}{
	Let $f(y)\propto \sum_{n=0}^{\infty} (-1)^n A_n (2\pi y^3)^{-0.5}\exp\{-(2n+b)^2(8y)^{-1}-0.5c^2y\}$ indicate the probability density function of a P\'{o}lya-Gamma distributed random variable $y \sim \text{PG}(b,c)$.
	For each $h=1, \ldots, H$, generate $p$ independent random variables $d_{j(h)}$ sampling from the full conditional distribution $(d_{j(h)}\mid -) \sim \text{PG}(1, x_{j}^\T \beta_h)$.
	Let $D_{(h)}$ denote the $p\times p$ diagonal matrix with entries $d_{j(h)}$ ($j=1,\ldots,p$).
}\Substep{3}{3}{
	Define the $q \times q$ diagonal matrix $B=\sigma^2_\beta I_q$.
	For each $h=1, \ldots, H$, update $\beta_{h}$ sampling from
	\begin{equation*}
	(\beta_{h} \mid -) \sim N_{q}\{(x^\T D_{(h)} x + B^{-1} )^{-1} (x^\T \kappa_h), \, (x^\T D_{(h)} x + B^{-1} )^{-1}\},
	\end{equation*}
	where $\kappa_{h}$ is a $p$-dimensional vector with the $j$-th entry equal to $\phi_{jh}^{(L)}-0.5$.
}
\Step{4}{
	Update the elements $\lambda_{jh}^*$ by sampling from the independent full conditional posterior distributions of the row vectors $\lambda_{j}^*=(\lambda_{j1}^*, \ldots, \lambda_{jH}^*)$, for $j=1,\ldots,p$,
	\begin{equation*}
	(\lambda_{j}^*\mid-) \sim N_{H}\big\{(D^{-1}+\sigma_j^{-2}\eta_{(j)}^{\T}\eta_{(j)})^{-1}\sigma_j^{-2} \eta_{(j)}^{\T} y^{(j)}, \, (D^{-1}+\sigma_j^{-2}\eta_{(j)}^{\T}\eta_{(j)})^{-1}\big\},
	\end{equation*} 
	where $\eta_{(j)}$ is the $n \times H$ matrix such that the generic element is $\eta_{(j) ih}=\eta_{ih} \sqrt{\rho_h} \sqrt{\phi_{jh}}$, $D^{-1}=\text{diag}(\vartheta_{1}^{-1},\ldots,\vartheta_{H}^{-1})$ and $y^{(j)}=(y_{1j}, \ldots, y_{nj})^\T$. Set $\lambda_{jh}= \lambda_{jh}^* \sqrt{\rho_h}\sqrt{\phi_{jh}}$.
}
\Step{5}{
	Update the column scales $\gamma_h$ (for $h=1, \ldots, H$), following the substeps below and setting $\gamma_h = \vartheta_h \rho_h$.
	Consistently with \citet{Legramanti2020}, define the independent indicators $z_h$ ($h=1,\ldots,p$) with prior $\text{pr}(z_h=l)=w_l$.
}
\Substep{5}{1}{
	Update the augmented data $z_h$ by sampling from the full conditional distribution
	\begin{equation}
	\label{eq:z-fc}
	\text{pr}(z_h=l) \propto 
	\begin{cases}
	w_l \,\prod_{i=1}^n \prod_{j=1}^p N(y_{ij}; \mu_{ijh}^{(0)} ,\sigma_j^2)
	\qquad \qquad \text{for} \quad l=1,\ldots,h \\
	w_l \,\prod_{i=1}^n \prod_{j=1}^p N(y_{ij}; \mu_{ijh}^{(1)} ,\sigma_j^2) \qquad \qquad \text{for} \quad l=h+1,\ldots,H, \\
	\end{cases}
	\end{equation}
	where $N(x; \mu,\sigma^2)$ indicates the Gaussian probability density function with mean $\mu$ and variance $	\sigma^2$. The mean values $\mu_{ijh}^{(0)}$ and $\mu_{ijh}^{(1)}$ are defined according to $\mu_{ijh}^{(z)}= \sum_{l \neq h}^{H} \sqrt{\rho_{l}}\,\sqrt{\phi_{jl}}\lambda_{jl}^* \eta_{il} + \sqrt{z}\,\sqrt{\phi_{jh}}\lambda_{jh}^* \eta_{ih}$.
	Set $\rho_h=1$ if $z_h >h$, else $\rho_h=0$.
}
\Substep{5}{2}{ 
	For $h=1,\ldots,H$, update $\vartheta_h^{-1}$ sampling from
	$ \text{Ga}(a_\theta+0.5p, b_\theta+0.5\sum_{j=1}^{p} \lambda_{jh}^{*\,2})$.
}
\Substep{5}{3}{
	For $l=1,\ldots,H-1$, sample $v_l$ from 
	\begin{equation*}
	(v_l\mid-) \sim \text{Be}\big\{1+ \sum_{h=1}^{H} \mathbbm{1}{(z_h=l)}, \alpha + \mathbbm{1}{(z_h>l)} \big\},
	\end{equation*}
	set $v_H = 1$ and update $w_l=v_l \prod_{m=1}^{l-1} (1-v_m)$, for $l=1,\ldots,H$.
}
\Step{6}{
	Update independently the local scales, for $j=1,\ldots,p$ and $h=1,\ldots,H$, by sampling from the full conditional distributions
	\begin{equation*}
	\text{pr}(\phi_{jh}=u) \propto \begin{cases}
	\{1-\text{logit}^{-1}(x_{j}^\T \beta_h)\, c_p\}\, \prod_{i=1}^n N(y_{ij}; \mu_{ijh}^{(u)} ,\sigma_j^2)
	\quad  \; \text{for} \; u=0\\
	\text{logit}^{-1}(x_{j}^\T \beta_h) \,c_p \prod_{i=1}^n N(y_{ij}; \mu_{ijh}^{(u)} ,\sigma_j^2)  \qquad \qquad \text{for} \;  u=1.
	\end{cases} 
	\end{equation*}
	with $\mu_{ijh}^{(u)}=\sum_{l \neq h}^{H} \sqrt{\rho_{l}}\,\sqrt{\phi_{jl}}\lambda_{jl}^* \eta_{il} + \sqrt{\rho_{h}}\,\sqrt{u}\lambda_{jh}^* \eta_{ih}$.
}

\subsection{Simulation settings}

The results reported in Section 4 are obtained running the algorithms for 25000 iterations discarding the first 10000 iterations. Then, we thin the Markov chain, saving every $5$-th sample.
We adapt the number of active factors at iteration $t$ with probability $p(t) = \exp(-1 -5\, 10^{-4} t)$. We set $a_\sigma =1$ and $b_\sigma=0.3$.
In the structured increasing shrinkage algorithm, we choose the offset constant $c_p=2e \log(p)/p$ which belongs to $(0,1)$ for every $p \geq 15$. 

In scenario $d$, the meta covariates in matrix $x_0$ are a categorical variable with four balanced categories, a continuous variable sampled from a multivariate Gaussian distribution, and a continuous variable where the $p$ elements are sampled from $p$ gamma distributions.

To infer the structural zeros within each column of $\Lambda$ in the cumulative shrinkage process and in the multiplicative gamma process, we set $\lambda_{jh}$ to zero when $|\lambda_{jh}|$ ($j =1,\ldots,p$) is under a certain threshold. We choose the threshold equal to 0.05, which is consistent with the value of the hyperparameter $\theta_\infty$ used in the cumulative shrinkage process.

To address column order ambiguity and label switching, we compute the mean classification error only after having ordered the columns of $\Lambda^{(t)}$ (for $t=1, \ldots,S$), for each model, increasingly with respect to the number of zero entries identified.

\subsection{Simulation results}
We report additional results for the simulation study of Section 4 of the main paper. 
\begin{table}[h]
	\def~{\hphantom{0}}
	\setlength{\tabcolsep}{12pt}
	\caption{Median and interquartile range of the LPML, Cov. MSE and of $E(H_a\mid y)$ computed in 25 replications assuming Scenario b and several combinations of $(p,k,s)$}{
		\begin{tabular}{lccccccc}
			& $(p,k,s)$ & \multicolumn{2}{c}{MGP} & \multicolumn{2}{c}{CUSP} & \multicolumn{2}{c}{SIS} \\
			& & Q$_{0. 5}$ & IQR & Q$_{0.5}$ & IQR & Q$_{0. 5}$ & IQR \\
			LPML  & ~~(16,4,0.6) &~-28.20 & 0.33 & ~-28.20 & 0.33 & ~-28.17 & 0.32 \\
			&  ~~(32,8,0.4)& ~-56.95 & 0.53 & ~-57.00 & 0.51 & ~-56.80 & 0.49\\ 
			& ~(64,12,0.3) & -111.35 & 0.70 & -111.71 & 0.74 & -110.76 & 0.89\\
			& (128,16,0.2)&-211.65 & 0.74 & -215.94 & 1.57 & -210.19 & 0.86 \\ 	
			Cov. MSE & ~~(16,4,0.6) & ~~0.25 & 0.12 & ~~0.25 & 0.12 & ~~0.23 & 0.10\\
			&  ~~(32,8,0.4) & ~~0.32 & 0.08 & ~~0.33 & 0.10 & ~~0.30 & 0.12\\
			& ~(64,12,0.3) & ~~0.37 & 0.10 & ~~0.43 & 0.11 & ~~0.22 & 0.09\\
			&(128,16,0.2) & ~~0.23 & 0.03 & ~~0.32 & 0.04 & ~~0.09 & 0.01\\ 
			$E(H_a\mid y)$ & ~~(16,4,0.6) & ~~8.91 & 1.52 & ~~4.00 & 0.00 & ~~4.00 & 0.00\\
			&  ~~(32,8,0.4) & ~11.27 & 1.48 & ~~7.00 & 1.00 & ~~8.00 & 0.00\\ 
			& ~(64,12,0.3) & ~14.72 & 1.49 & ~11 .00 & 0.00 & ~12.00 & 0.00 \\
			&(128,16,0.2) & ~17.16 & 0.81 & ~12.00 & 1.75  & ~16.00 & 0.00  \\
	\end{tabular}}
	\label{tab:ScenarioB}       
	{\footnotesize
		LPML, logarithm of the pseudo-marginal likelihood; Cov. MSE, covariance mean squared error;	CUSP, cumulative increasing shrinkage process; MGP, multiplicative gamma process; SIS, structured increasing shrinkage process; 
		Q$_{0.5}$, median; IQR, interquartile range.
	}
\end{table}

In scenario b, we also apply the method proposed by \citet{rockova2016}, which is referred to as parameter expanded likelihood expected maximization.  Hyperparameters are set as suggested by the authors. This approach focuses on finding a sparse mode based on an over-parameterized factor model.
The performance in terms of mean squared error in covariance estimation and classification error in detecting sparsity in $\Lambda$ is reported in Table S2.   The results are not competitive with the other approaches we have considered.

\begin{table}[h]
	\def~{\hphantom{0}}
	\setlength{\tabcolsep}{15pt}
	\caption{Median and interquartile range of the Cov. MSE and CE computed applying the parameter expanded likelihood expected maximization method of \citet{rockova2016} in 25 replications under Scenario b and several combinations of $(p,k,s)$}{
		\begin{tabular}{ccccc}
			$(p,k,s)$ & \multicolumn{2}{c}{Cov. MSE} & \multicolumn{2}{c}{CE} \\
			& Q$_{0. 5}$ & IQR & Q$_{0.5}$ & IQR\\
			~~(16,4,0.6) &0.55 & 0.16 & 0.77 & 0.14\\
			~~(32,8,0.4)& 0.66 & 0.11 & 0.76 & 0.10 \\ 
			~(64,12,0.3) & 0.64 & 0.11 & 0.90 & 0.10 \\
			(128,16,0.2)& 0.42 & 0.03 & 1.06 & 0.20  \\ 	
	\end{tabular}}
	\label{tab:ScenarioB}       
	{\footnotesize
		Cov. MSE, covariance mean squared error; CE, classification error;	PXLEM, parameter expanded likelihood expected maximization; Q$_{0.5}$, median; IQR, interquartile range.
	}
\end{table}

\begin{table}[h]
	\def~{\hphantom{0}}
	\setlength{\tabcolsep}{12pt}
	\caption{Median and interquartile range of the LPML, Cov. MSE, $E(H_a\mid y)$ and MCE computed in 25 replications assuming Scenario c and several combinations of $(p,k,s)$}{
		\begin{tabular}{lccccccc}
			& $(p,k,s)$ & \multicolumn{2}{c}{MGP} &  \multicolumn{2}{c}{CUSP} & \multicolumn{2}{c}{SIS} \\
			& & Q$_{0.5}$ & IQR & Q$_{0. 5}$ & IQR & Q$_{0. 5}$ & IQR  \\
			LPML  & ~~(16,4,0.6) & ~-27.62 & 0.24& ~-27.62 & 0.25 & ~-27.59 & 0.24 \\
			& ~~(32,8,0.4)& ~-56.16 & 0.64  & ~-56.22 & 0.51  & ~-55.89 & 0.59 \\ 
			& ~(64,12,0.3) &-109.64 & 0.69 &-110.67 & 0.71 &-109.06 & 0.65 \\
			& (128,16,0.2) &-209.57 & 0.88   & -214.19 & 1.76  & -208.34 & 1.04 \\ 
			Cov. MSE  & ~~(16,4,0.6) & ~~0.30 & 0.10 & ~~0.29 & 0.09 & ~~0.26 & 0.11 \\
			& ~~(32,8,0.4)& ~~0.77 & 0.26 & ~~0.72 & 0.18 & ~~0.80 & 0.43\\ 
			& ~(64,12,0.3) & ~~1.01 & 0.35 & ~~0.94 & 0.21 & ~~1.20 & 1.22 \\
			& (128,16,0.2) & ~~0.78 & 0.18 & ~~0.87 & 0.21 & ~~0.35 & 0.48 \\
			$E(H_a\mid y)$ & ~~(16,4,0.6)& ~~8.38 & 1.80  & ~~3.44 & 1.00  & ~~4.00 & 0.00\\
			& ~~(32,8,0.4) & ~10.38 & 1.12  & ~~5.05 & 0.91  & ~~8.00 & 1.00\\ 
			& ~(64,12,0.3) & ~13.67 & 1.20  & ~~8.00 & 0.92  & ~12.00 & 0.00 \\
			& (128,16,0.2) & ~16.56 & 0.83  & ~~9.00 & 0.00  & ~16.00 & 0.00 \\ 
			MCE & ~~(16,4,0.6) & ~~0.98 & 0.17  &  ~~0.53 & 0.20 & ~~0.24 & 0.06  \\
			& ~~(32,8,0.4) & ~~0.65 & 0.07 &  ~~0.44 & 0.08 &  ~~0.19 & 0.07 \\ 
			& ~(64,12,0.3) & ~~0.59 & 0.04 &~~0.48 & 0.04 & ~~0.18 & 0.06\\
			& (128,16,0.2) &~~0.48 & 0.02 & ~~0.44 & 0.01&  ~~0.06 & 0.10 \\  
	\end{tabular}}
	\label{tab:ScenarioC}       
	{\footnotesize
		LPML, logarithm of the pseudo-marginal likelihood; Cov. MSE, covariance mean squared error; MCE, mean classification error;
		CUSP, cumulative increasing shrinkage process; MGP, multiplicative gamma process; SIS, structured increasing shrinkage process; Q$_{0.5}$, median;  IQR, interquartile range.
	}
\end{table}

\begin{table}[h]
	\setlength{\tabcolsep}{8pt}
	\def~{\hphantom{0}}
	\caption{Median and interquartile range of the LPML, Cov. MSE, $E(H_a\mid y)$ and MCE computed in 25 replications assuming Scenario d and several combinations of $(p,k,s)$}{
		\begin{tabular}{lccccccccc}
			& $(p,k,s)$ & \multicolumn{2}{c}{MGP} &  \multicolumn{2}{c}{CUSP} &  \multicolumn{2}{c}{SIS} &  \multicolumn{2}{c}{SIS$_{mc}$} \\
			& & Q$_{0.5}$ & IQR & Q$_{0.5}$ & IQR & Q$_{0\cdot 5}$ & IQR & Q$_{0.5}$ & IQR \\
			LPML  & ~~(16,4,0.6) & ~-27.74 & 0.43 & ~-27.75 & 0.43 & ~-27.71 & 0.40 & ~-27.73 & 0.40\\
			& ~~(32,8,0.4)& ~-56.25 & 0.69 & ~-56.35 & 0.72 & ~-56.16 & 0.68 & ~-56.12 & 0.65\\ 
			& ~(64,12,0.3) & -109.72 & 0.61 & -110.54 & 0.88 & -109.27 & 0.46 & -109.16 & 0.57 \\
			& (128,16,0.2) & -209.60 & 0.48 & -213.50 & 1.21 & -208.11 & 0.42 & -208.03 & 0.47 \\  	
			Cov. MSE & ~~(16,4,0.6)& ~~0.31 & 0.11 & ~~0.30 & 0.14 & ~~0.28 & 0.14 & ~~0.27 & 0.16\\
			& ~~(32,8,0.4) & ~~0.70 & 0.25 & ~~0.71 & 0.18 & ~~0.75 & 0.22 & ~~0.79 & 0.78\\ 
			& ~(64,12,0.3) & ~~1.03 & 0.43 & ~~0.91 & 0.29 & ~~1.51 & 0.59 & ~~1.16 & 0.84 \\
			& (128,16,0.2) & ~~0.93 & 0.49 & ~~0.90 & 0.33 & ~~1.49 & 1.21 & ~~1.28 & 1.81 \\ 
			$E(H_a\mid y)$ & ~~(16,4,0.6)& ~~8.60 & 0.64 & ~~3.96 & 0.80 & ~~4.00 & 0.00 & ~~4.00 & 0.00\\
			& ~~(32,8,0.4) & ~10.71 & 1.24 & ~~5.75 & 1.00 & ~~7.00 & 1.00 & ~~8.00 & 0.00\\ 
			& ~(64,12,0.3) & ~13.93 & 1.37 & ~~8.00 & 0.92 & ~12.00 & 0.00 & ~12.00 & 0.00 \\
			& (128,16,0.2) & ~16.56 & 0.88 & ~~9.00 & 0.00 & ~16.00 & 1.00 & ~16.00 & 1.00 \\ 
			MCE & ~~(16,4,0.6) & ~~0.94 & 0.13 & ~~0.64 & 0.19 & ~~0.26 & 0.08 & ~~0.23 & 0.13  \\
			& ~~(32,8,0.4) & ~~0.67 & 0.10 & ~~0.49 & 0.09 & ~~0.20 & 0.08 & ~~0.20 & 0.10\\ 
			& ~(64,12,0.3) & ~~0.58 & 0.05 & ~~0.47 & 0.04 & ~~0.21 & 0.06 & ~~0.21 & 0.08 \\
			& (128,16,0.2) &~~0.49 & 0.02 & ~~0.43 & 0.02 & ~~0.18 & 0.11 & ~~0.17 & 0.11\\  
	\end{tabular}}
	\label{tab:ScenarioD}       
	{\footnotesize
		LPML, logarithm of the pseudo-marginal likelihood; Cov. MSE, covariance mean squared error; MCE, mean classification error;
		CUSP, cumulative increasing shrinkage process; MGP, multiplicative gamma process; SIS, structured increasing shrinkage process;  SIS$_{mc}$, structured increasing shrinkage process with meta covariates;  Q$_{0. 5}$, median; IQR, interquartile range.
	}
\end{table}

\begin{figure}[h]
	\centering
	\subfigure{\includegraphics[width=.42\textwidth]{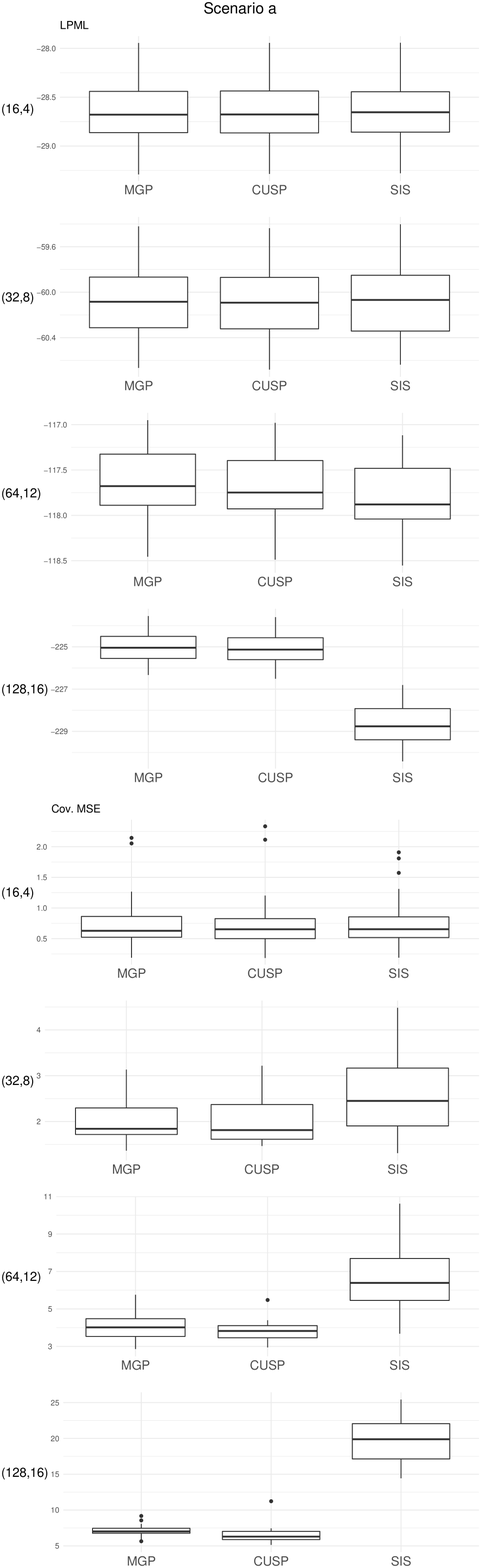}} 
	\subfigure{\includegraphics[width=.42\textwidth]{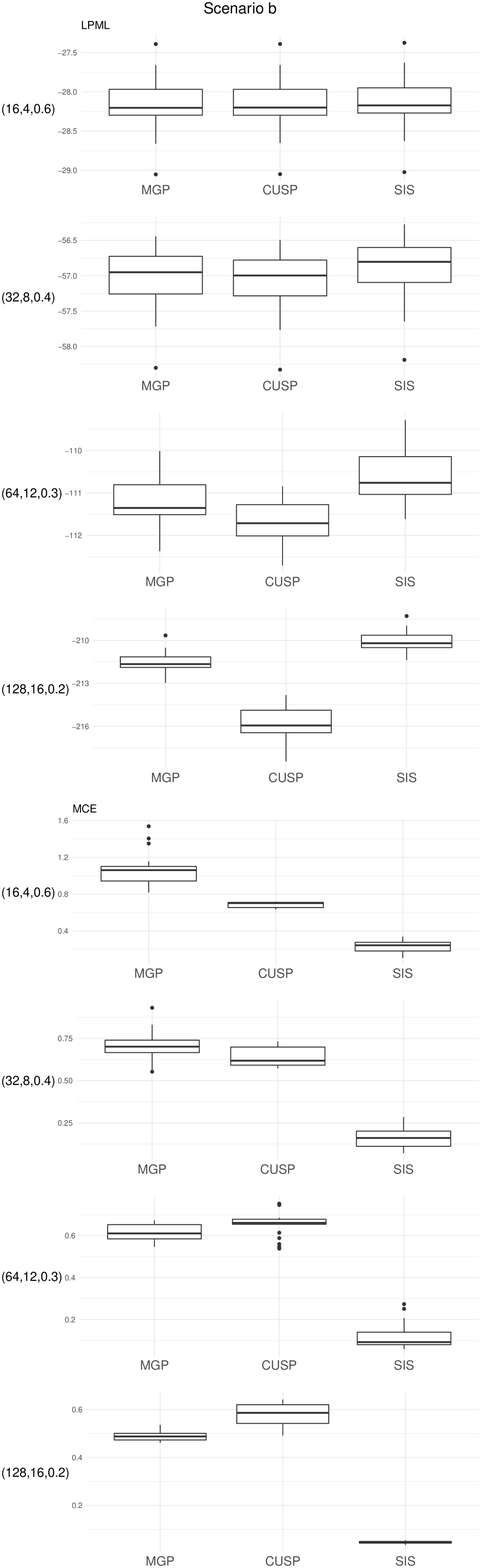}}
	\caption{Boxplots of the logarithm of the pseudo-marginal likelihood for all combinations $(p,k)$ in scenario a (top left panel) and scenario b (top right panel), of the covariance mean square error in scenario a (bottom left panel), and of the 
		mean classification error in scenario b (bottom right panel). 
		LPML, logarithm of the pseudo-marginal likelihood; Cov. MSE, covariance mean squared error; MCE, mean classification error; CUSP, cumulative increasing shrinkage process; MGP, multiplicative gamma process; SIS, structured increasing shrinkage process.}
	\label{fig:sm-1}
\end{figure}

\begin{figure}[h]
	\centering
	\subfigure{\includegraphics[width=.42\textwidth]{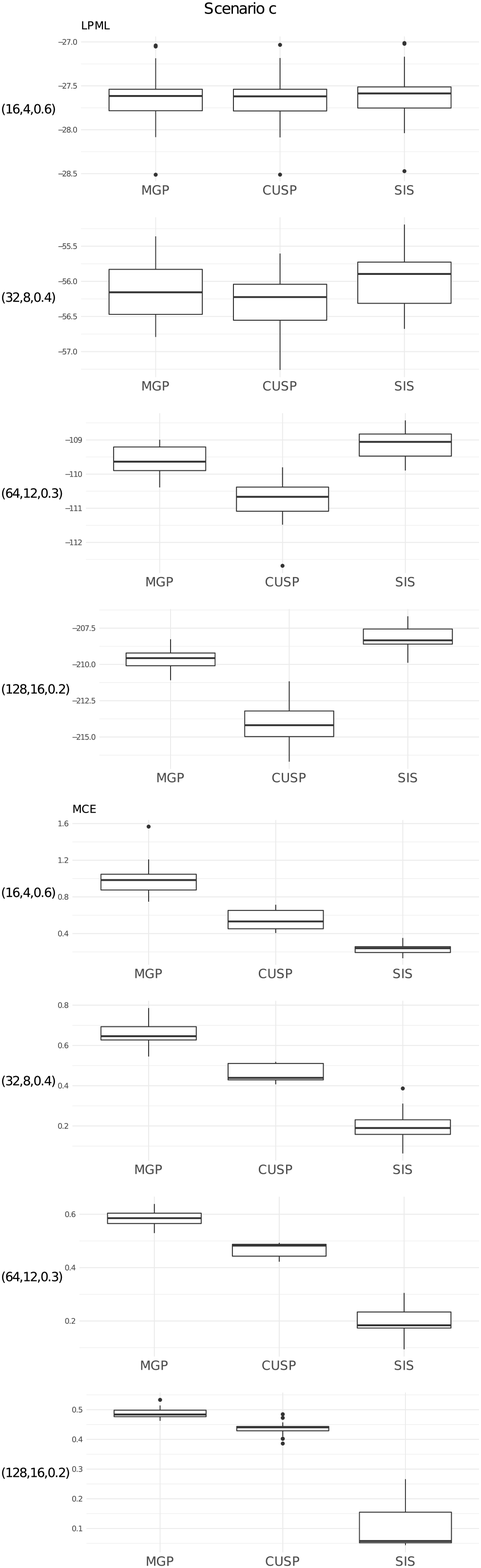}} 
	\subfigure{\includegraphics[width=.42\textwidth]{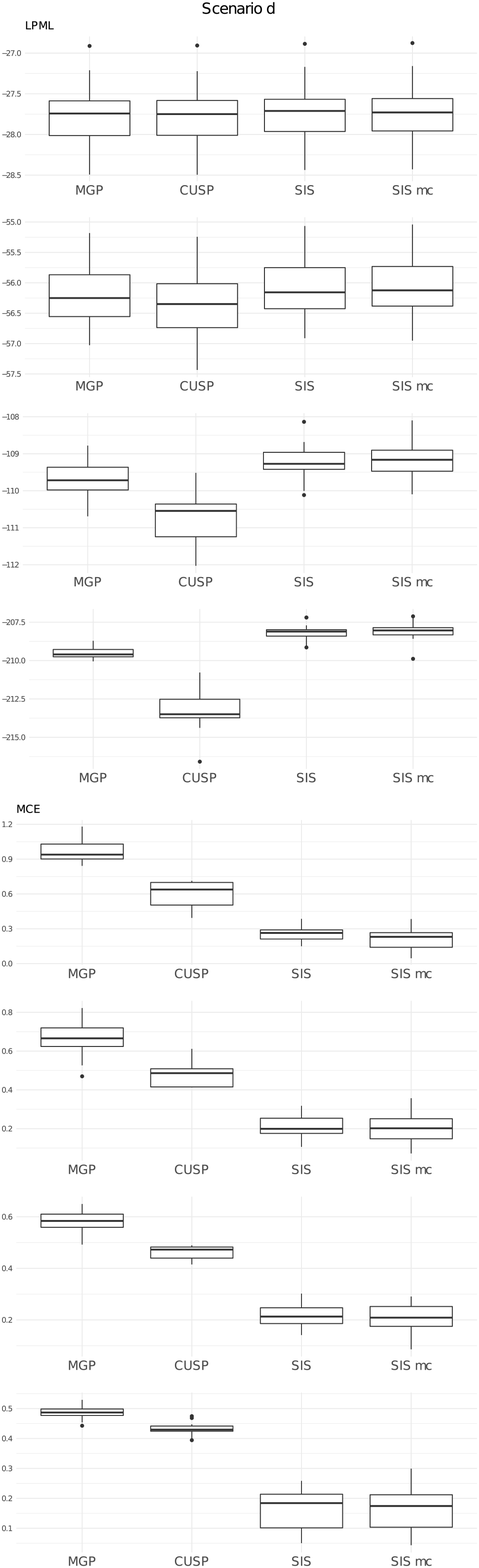}}
	\caption{Boxplots of the logarithm of the pseudo-marginal likelihood and of the mean classification error of each model for all combinations of $(p,k,s)$ in Scenario c (left panel) and Scenario d (right panel). LPML, logarithm of the pseudo-marginal likelihood; MCE, mean classification error; CUSP, cumulative shrinkage process; MGP, multiplicative gamma process; SIS, structured increasing shrinkage process; SIS mc, structured increasing shrinkage process with meta covariates.}
	\label{fig:sm-2}
\end{figure}

\clearpage

\subsection{Simulation study of sensitivity to hyperparameters and truncation level}
We conduct further simulation experiments to assess the impact of some hyperparameters on relevant prior and posterior summaries.
Figure \ref{fig:sm-priorprop} displays the prior distribution, obtained simulating 10{,}000 samples from the prior, of the proportion of variance explained by the structured increasing shrinkage factor model for varying $\alpha$,  $\{E(\sigma^{-2}),\text{var}(\sigma^{-2})\}$, and  $\{E(\vartheta_h^{-1}),\text{var}(\vartheta_h^{-1})\}$.
The hyperparameter $\alpha$, representing the expected number of factors, positively affects the proportion of explained variance. The influence of the hyperparameters regulating the distribution of $\vartheta_h$ is even clearer, with concentrated prior on a large value of $E(\vartheta^{-1})$,  inducing a smaller proportion of variance explained by the factor model. The role of $\{E(\sigma^{-2}),\text{var}(\sigma^{-2})\}$ is less clear, but suggests that sufficiently large mean and variance can guarantee higher flexibility.

\begin{figure}[h]
	\centering
	\includegraphics[width=.85\textwidth]{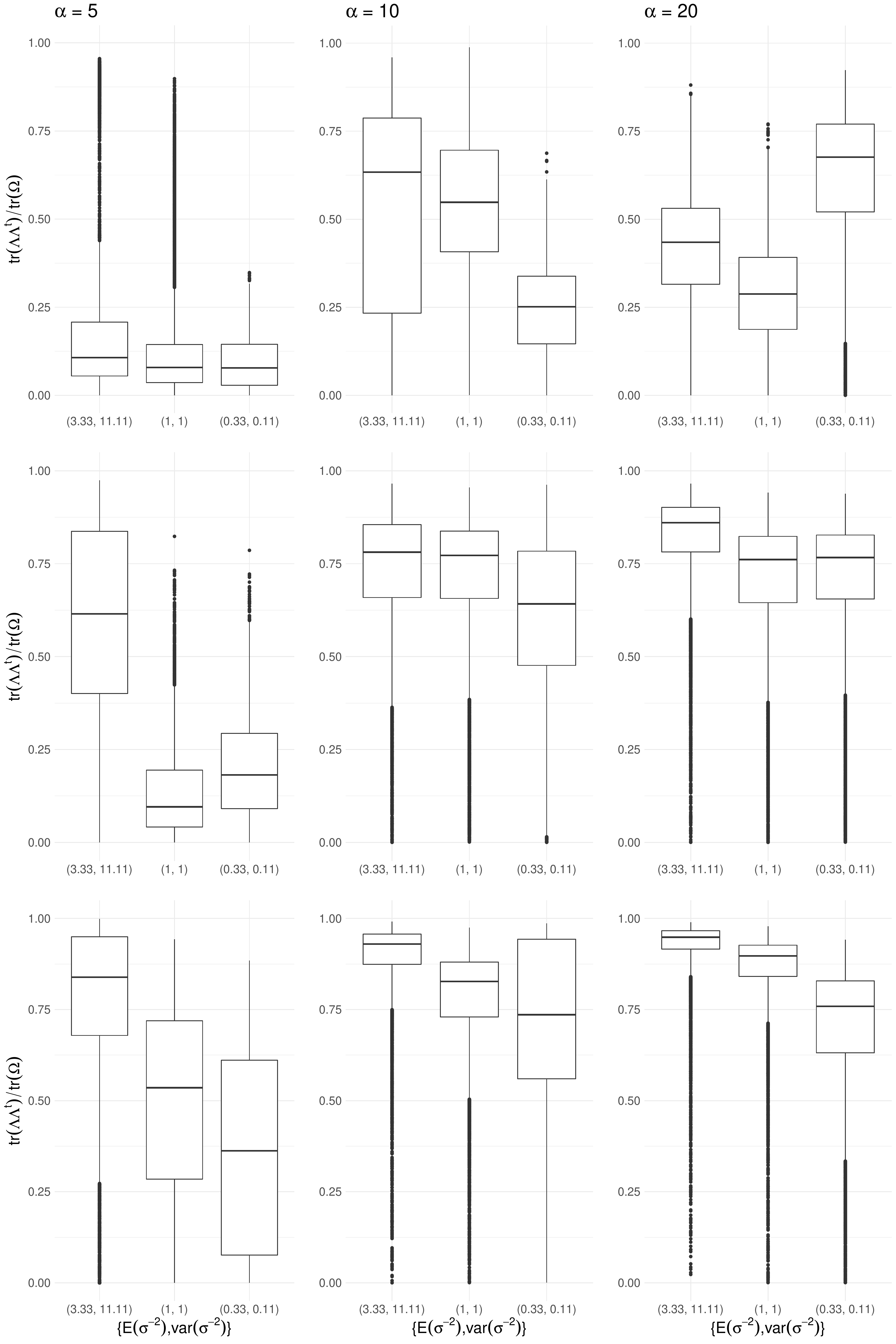}
	\caption{Boxplots of the prior proportion of variance explained by the factor model $\text{tr}(\Lambda \Lambda^\T)/\text{tr}(\Omega)$. The quantity is obtained simulating 10{,}000 samples from the prior distribution with varying values of the parameters. The horizontal axis characterize the effect of  $\{E(\sigma^{-2}),\text{var}(\sigma^{-2})\}$; differences for $\alpha \in (5,10,20)$ are reported in each column; differences for $\{E(\vartheta^{-1}), \text{var}(\vartheta^{-1})\} \in \{(2,2),(1, 0.5),(0.5, 0.125)\}$ are reported in each  row.}
	\label{fig:sm-priorprop}
\end{figure}

The latter comment is confirmed by the study of the impact of $\alpha$ and $\{E(\sigma^{-2}),\text{var}(\sigma^{-2})\}$ on the posterior bound of the truncation error and on the posterior distribution of the proportion of variance explained by the factor model. Specifically, we generate synthetic data sets with $n=100$ observations with dimension $p=50$ from the Gaussian linear factor model $y_i \sim N_p(0, \Lambda_0 \Lambda_0^T + I_p) $, with $\Lambda_0$ a sparse $p\times k$ matrix with $k=50$. We randomly set two thirds of the elements of $\Lambda_0$ equal to zero, drawing the non zero elements from a Gaussian distribution with mean zero and variances $\theta_{h}$ sampled from an inverse gamma distribution $\theta_h^{-1}\sim \text{Ga}(2,2)$.  We keep the number of active factors $H$ fixed at $50$, and set $a_\theta=b_\theta=2$ and $c_p=2e \log(p)/p$. 
We run the Gibbs algorithm for the structured increasing shrinkage model for 15000 iterations, discarding the first 5000 iterations. Then, we thin the Markov chain, saving every 5-th sample.

In Figure \ref{fig:sm-proportion} the sampled posterior distribution of the proportion of variance explained by the factor model $\text{tr}(\Lambda \Lambda^\T)/\text{tr}(\Omega)$ is reported for varying $\alpha$ and $\{E(\sigma^{-2}), \text{var}(\sigma^{-2})\}$. The same proportion computed on the matrices generating the data is $\text{tr}(\Lambda_0 \Lambda_0^\T)/\text{tr}(\Omega_0)=0.966$, with $\Omega_0 = \Lambda_0 \Lambda_0^\T+ I_{50}$. A sufficiently concentrated prior on a large value of $E(\sigma^{-2})$ seems more suitable to model such data, even if we have incorrect expectations on the number of factors, i.e. $\alpha$ set small.

\begin{figure}[h]
	\centering
	\includegraphics[width=.95\textwidth]{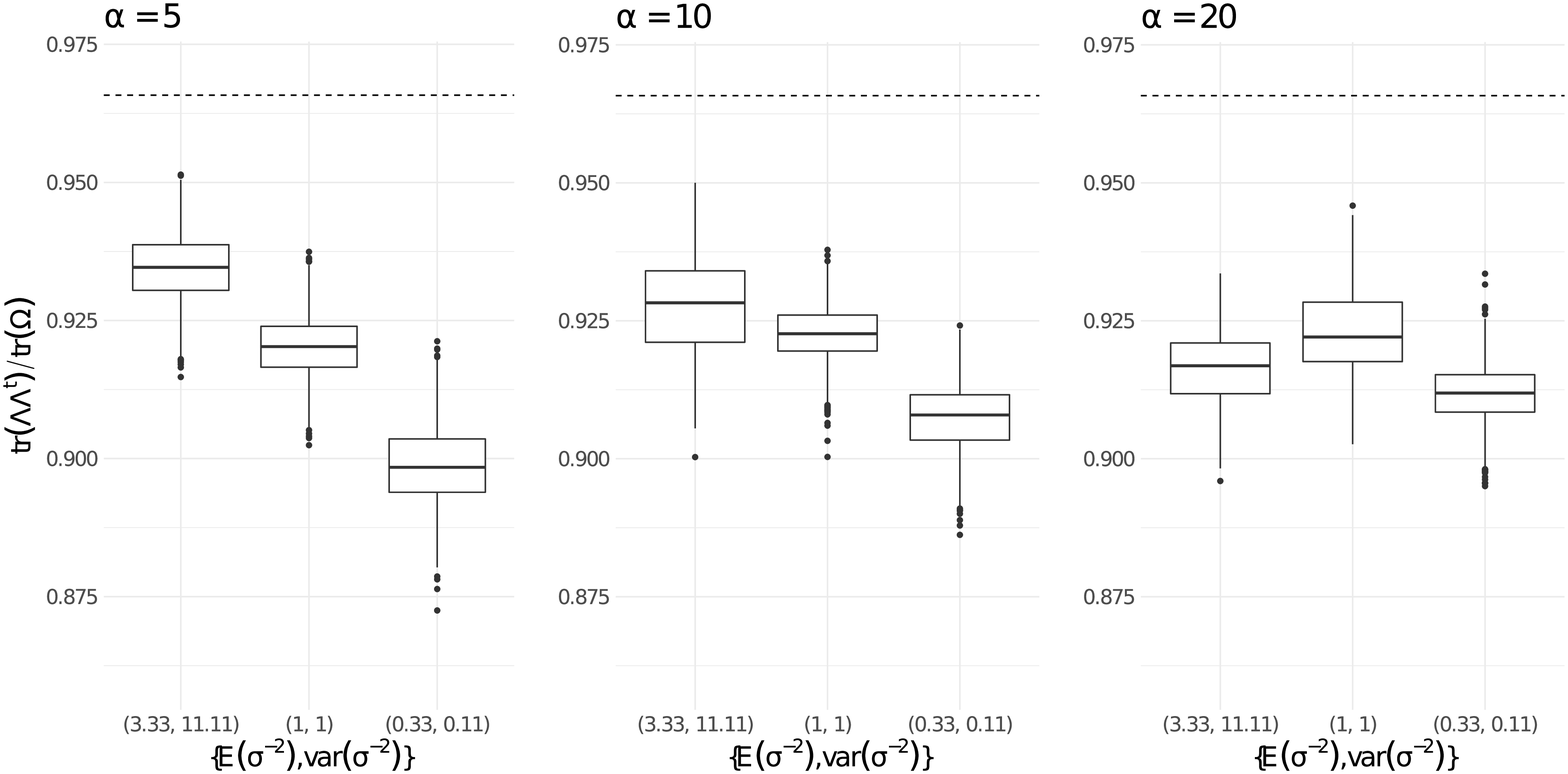}
	\caption{Boxplots representing the simulated posterior distribution of the proportion of variance explained by the factor model $\text{tr}(\Lambda \Lambda^\T)/\text{tr}(\Omega)$ for varying $\alpha$ and $\{E(\sigma^{-2}),\text{var}(\sigma^{-2})\}$. The dashed lines represent the proportion computed on the true value of $\Lambda$ and $\Omega$.}
	\label{fig:sm-proportion}
\end{figure}

Figure \ref{fig:sm-trunc-err} displays the Monte Carlo approximation of the posterior probability of truncation error $\text{pr}\{\text{tr}(\Omega_H)/\text{tr}(\Omega)<T\}$ for different values of $H$ and $T$ and varying $\alpha$ and $\{E(\sigma^{-2}), \text{var}(\sigma^{-2})\}$. If $\Lambda_0$ is sparse, a small value of $\alpha$ induces good approximations even with $H$ smaller than the true number of factors. The inferred sparsity pattern in $\Lambda$ is robust to the prior distribution for 
$\sigma^{-2}$.

\clearpage
\begin{figure}[h]
	\centering
	\includegraphics[width=.95\textwidth]{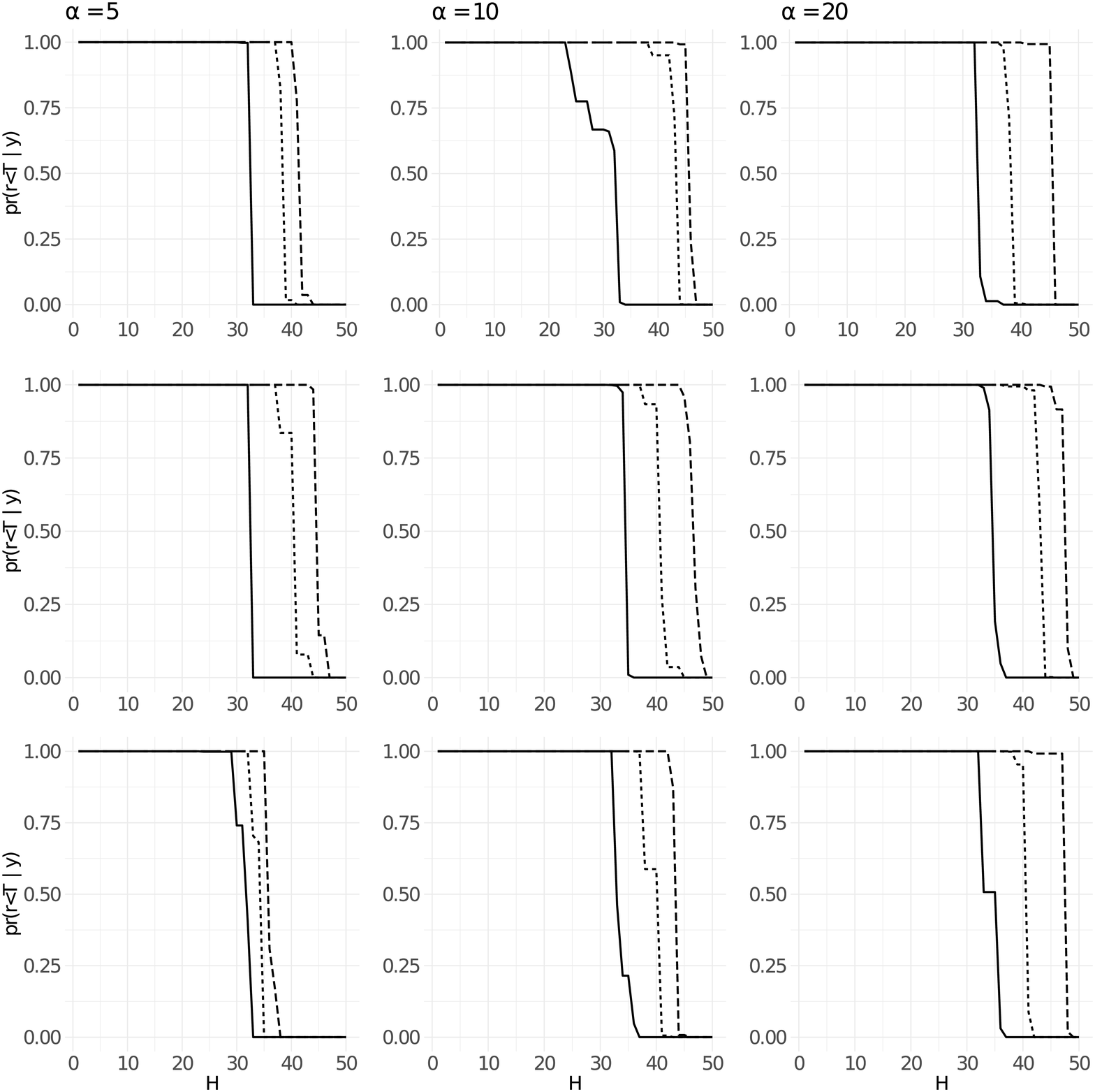}
	\caption{Monte Carlo approximation of  the posterior probability of truncation error $\text{pr}(r<T \mid y)$, with $r=\text{tr}(\Omega_H)/\text{tr}(\Omega)$, at varying of $H$. The quantity is computed for $T$ equal to $0.75$ (---), $0.9$ (- - -), and $0.95$ (-- -- --) and varying $\alpha \in (5,10,20)$ over the columns and $\{E(\sigma^{-2}), \text{var}(\sigma^{-2})\} \in \{(3.33, 11.11),(1, 1),(0.33, 0.11)\}$ over the rows of the figure.}
	\label{fig:sm-trunc-err}
\end{figure}

\section{Finnish bird co-occurrence application}

\subsection{Gibbs algorithms of probit structured increasing shrinkage model}

In case of probit data (see Section 5 of the main paper) and  the structured increasing shrinkage process, we can rewrite the latent model for $z_{ij}$ as 
\begin{align*}
z_{ij}&= w_{i}^T \mu_{j} + \epsilon_{ij},\\
\epsilon_{ij}&= \sum_{h=1}^\infty \sqrt{\rho_{h}} \sqrt{\phi_{jh}} \,  \lambda_{jh}^* \eta_{ih} + \varepsilon_{ij}, \qquad \lambda_{jh}^* \sim N(0, \vartheta_h), \qquad  \varepsilon_{ij}\sim N(0, 1),
\end{align*}
where $\lambda_{jh}^*$ is an absolutely continuous random variable.
Let the notation $(x\mid-)$ denote the full conditional distribution of $x$ conditionally on everything else.
Given $H$ the number of factors of the truncated model, the sampler cycles through the following steps.
\Step{1}{
	Update $\mu_j$, for every $j=1, \ldots,p,$ by sampling from the independent full conditional posterior distributions 
	\begin{equation*}
	(\mu_{j} \mid -) \sim N_c \big[ 
	(\sigma_\mu^{-2} I_c + w^\T w)^{-1} \{w^\T (z^{(j)}-\eta\lambda_{j} )+ b x_j\},(\sigma_\mu^{-2} I_c + w^\T w)^{-1}\big],
	\end{equation*}
	where $z^{(j)}=(z_{1j}, \ldots, z_{nj})^\T$ and $\eta=(\eta_{1}, \ldots, \eta_{n})^\T$.
}
\Step{2}{
	Update $b_{l}$ ($l=1,\ldots,c$) sampling from conditionally independent posteriors
	\begin{equation*}
	(b_{l}\mid-) \sim N_{q}\big\{
	(\sigma_b^{-2} I_q + \sigma_\mu^{-2} x^\T x)^{-1} \sigma_\mu^{-2}(x^\T \mu^{(l)}),\,(\sigma_b^{-2} I_q + \sigma_\mu^{-2} x^\T x)^{-1}\big\},
	\end{equation*}
	where $\mu^{(l)}=(\mu_{1l}, \ldots, \mu_{pl})^\T$
}
\Step{3}{
	Update the elements $z_{ij}$ ($i=1,\ldots,n$; $j=1\ldots,p$) sampling independently from the truncated normal
	\begin{equation*}
	(z_{ij} \mid -) \sim TN( \lambda_{j}^\T \eta_{i} + w_{i}^T \mu_{j},  1, l_{ij}, u_{ij}), 
	\end{equation*}
	where the lower bound $l_{ij}$ is equal to $0$ if $y_{ij}=1$ and $-\infty$ otherwise. The upper bound $u_{ij}= 0$ if $y_{ij}=0$ and $\infty$ otherwise. Then, set $\epsilon= z - w \mu$.
}
\Step{4}{
	Update, for $i=1, \ldots, n$, the factor $\eta_i$ according to the posterior full conditional 
	\begin{equation*}
	(\eta_{i}\mid-) \sim N_{H}\big\{(I_H+\Lambda_H^\T \Lambda_H)^{-1} \Lambda_H^\T \epsilon_{i},\, (I_H+\Lambda_H^\T \Lambda_H)^{-1}\big\}.
	\end{equation*} 
}
\Step{5}{
	Update $\beta_h$ ($h=1,\ldots,H$) exploiting  the P\'{o}lya-Gamma data-augmentation strategy \citep{polson2013} and the decompostition $\phi_{jh} =\phi_{jh}^{(L)} \phi_{jh}^{(C)}$, with $\phi_{jh}^{(L)} \phi_{jh}^{(C)}$ independent a priori and distributed as Ber$\{\text{logit}^{-1}(x_j^\T \beta_h)\}$ and Ber$(c_p)$, respectively.
}
\Substep{5}{1}{
	Update $\phi_{jh}^{(L)}$, for $j=1,\ldots,p$ and $h=1,\ldots,H$, setting $\phi_{jh}^{(L)}=1$ if $\phi_{jh}=1$ and sampling from the full conditional distribution
	\begin{equation*}
	\text{pr}(\phi_{jh}^{(L)}=l) \propto 
	\begin{cases}
	1-\text{logit}^{-1}(x_j^\T \beta_h)\qquad \qquad \qquad \quad \, \text{for} \;
	l=0, \\
	\text{logit}^{-1}(x_j^\T \beta_h)(1-c_p)  \qquad \,   \text{for} \;l=1,\\
	\end{cases}
	\end{equation*}		
	if $\phi_{jh}=0$.
}
\Substep{5}{2}{
	Let $f(y)\propto \sum_{n=0}^{\infty} (-1)^n A_n (2\pi y^3)^{-0.5}\exp\{-(2n+b)^2(8y)^{-1}-0.5c^2y\}$ indicate the probability density function of a P\'{o}lya-Gamma distributed random variable $y \sim \text{PG}(b,c)$.
	For each $h=1, \ldots, H$, generate $p$ independent random variables $d_{j(h)}$ sampling from the full conditional distribution $(d_{j(h)}\mid -) \sim \text{PG}(1, x_{j}^\T \beta_h)$.
	Let $D_{(h)}$ denote the $p\times p$ diagonal matrix with entries $d_{j(h)}$ ($j=1,\ldots,p$).
}\Substep{5}{3}{
	Define the $q \times q$ diagonal matrix $B=\sigma^2_\beta I_q$.
	For each $h=1, \ldots, H$, update $\beta_{h}$ sampling from
	\begin{equation*}
	(\beta_{h} \mid -) \sim N_{q}\{(x^\T D_{(h)} x + B^{-1} )^{-1} (x^\T \kappa_h), \, (x^\T D_{(h)} x + B^{-1} )^{-1}\},
	\end{equation*}
	where $\kappa_{h}$ is the $p$-dimensional vector with the $j$-th entry equal to $\phi_{jh}^{(L)}-0.5$.
}
\Step{6}{
	Update the elements $\lambda_{jh}^*$ by sampling from the independent full conditional posterior distributions of the rows vector $\lambda_{j}^*=(\lambda_{j1}^*, \ldots, \lambda_{jH}^*)$, for $j=1,\ldots,p$,
	\begin{equation*}
	(\lambda_{j}^*\mid-) \sim N_{H}\big\{(D^{-1}+\eta_{(j)}^{\T}\eta_{(j)})^{-1} \eta_{(j)}^{\T} \epsilon^{(j)}, \, (D^{-1}+\eta_{(j)}^{\T}\eta_{(j)})^{-1}\big\},
	\end{equation*} 
	where $\eta_{(j)}$ is the $n \times H$ matrix such that the generic element is $\eta_{(j) ih}=\eta_{ih} \sqrt{\rho_h} \sqrt{\phi_{jh}}$, $D^{-1}=\text{diag}(\vartheta_{1}^{-1},\ldots,\vartheta_{H}^{-1})$ and $\epsilon^{(j)}=(\epsilon_{1j}, \ldots, \epsilon_{nj})^\T$. Set $\lambda_{jh}= \lambda_{jh}^* \sqrt{\rho_h}\sqrt{\phi_{jh}}$.
}
\Step{7}{
	Update the column scales $\gamma_h$ (for $h=1, \ldots, H$), following the substeps below and setting $\gamma_h = \vartheta_h \rho_h$.
	Consistently with \citet{Legramanti2020}, define the independent indicators $u_h$ ($h=1,\ldots,p$) with prior $\text{pr}(u_h=l)=w_l$.
}
\Substep{7}{1}{
	Update the augmented data $u_h$ by sequentially sampling from the full conditional distribution
	\begin{equation}
	\label{eq:z-fc}
	\text{pr}(u_h=l) \propto 
	\begin{cases}
	w_l \,\prod_{i=1}^n \prod_{j=1}^p N(\epsilon_{ij}; \mu_{ijh}^{(0)} ,\sigma_j^2)
	\qquad \qquad \text{for} \quad l=1,\ldots,h \\
	w_l \,\prod_{i=1}^n \prod_{j=1}^p N(\epsilon_{ij}; \mu_{ijh}^{(1)} ,\sigma_j^2) \qquad \qquad \text{for} \quad l=h+1,\ldots,H. \\
	\end{cases}
	\end{equation}
	The mean values $\mu_{ijh}^{(0)}$ and $\mu_{ijh}^{(1)}$ are defined according to $\mu_{ijh}^{(u)}= \sum_{l \neq h}^{H} \sqrt{\rho_{l}}\,\sqrt{\phi_{jl}}\lambda_{jl}^* \eta_{il} + \sqrt{u}\,\sqrt{\phi_{jh}}\lambda_{jh}^* \eta_{ih}$.
	Set $\rho_h=1$ if $u_h >h$, else $\rho_h=0$.
}
\Substep{7}{2}{ 
	For $h=1,\ldots,H$, update $\vartheta_h^{-1}$ sampling from
	$ \text{Ga}(a_\theta+0.5p, b_\theta+0.5\sum_{j=1}^{p} \lambda_{jh}^{*\,2})$.
}
\Substep{7}{3}{
	For $l=1,\ldots,H-1$, sample $v_l$ from 
	\begin{equation*}
	(v_l\mid-) \sim \text{Be}\big\{1+ \sum_{h=1}^{H} \mathbbm{1}{(u_h=l)}, \alpha + \mathbbm{1}{(u_h>l)} \big\},
	\end{equation*}
	set $v_H = 1$ and update $w_l=v_l \prod_{m=1}^{l-1} (1-v_m)$, for $l=1,\ldots,H$.
}
\Step{8}{
		Update the local scales, independently for $j=1,\ldots,p$ and sequentially over $h=1,\ldots,H$, by sampling from the full conditional distributions
	\begin{equation*}
	\text{pr}(\phi_{jh}=u) \propto \begin{cases}
	\{1-\text{logit}^{-1}(x_{j}^\T \beta_h)\,c_p\}\, \prod_{i=1}^n N(\epsilon_{ij}; \mu_{ijh}^{(u)} ,1)
	\quad  \; \text{for} \; u=0\\
	\text{logit}^{-1}(x_{j}^\T \beta_h) \,c_p \prod_{i=1}^n N(\epsilon_{ij}; \mu_{ijh}^{(u)} ,1)  \qquad \qquad \text{for} \;  u=1.
	\end{cases} 
	\end{equation*}
	with $\mu_{ijh}^{(u)}=\sum_{l \neq h}^{H} \sqrt{\rho_{l}}\,\sqrt{\phi_{jl}}\lambda_{jl}^* \eta_{il} + \sqrt{\rho_{h}}\,\sqrt{u}\lambda_{jh}^* \eta_{ih}$.
}

The results reported in Section 5 are obtained running the algorithm for 40000 iterations discarding the first 20000 iterations. Then, we thin the Markov Chain, saving every $5$-th sample. We adapt the number of active factors at iteration $t$ with probability $p(t) = \exp(-1 -2.5\, 10^{(-4)} t)$ and, given the high value of $p$ considered, we choose the offset constant $c_p=2e \log(p)/p$ which belongs to $(0,1)$ for every $p \geq 15$.

\clearpage
\subsection{Gibbs chains mixing}

\begin{figure}[h]
	\centering
	\includegraphics[width=.72\textwidth]{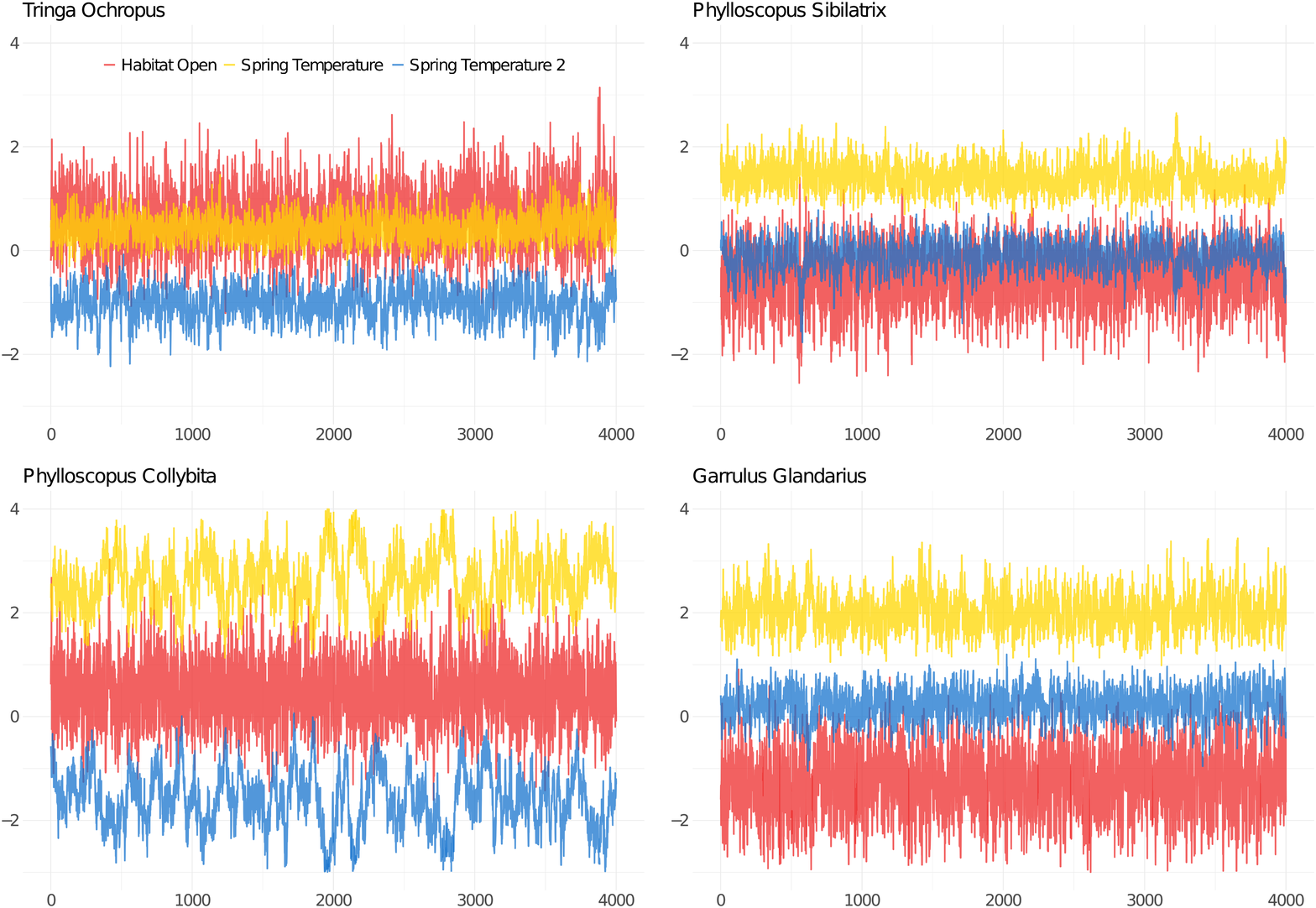}
	\caption{Chain plots of the marginal posterior samples of 12 mean coefficients of the matrix $\mu$ obtained by the Gibbs sampler, discarding the first 20000 iterations and saving every $5$-th sample.}
	\label{fig:sm-mixing1}
\end{figure}

\begin{figure}[h]
	\centering
	\includegraphics[width=.72\textwidth]{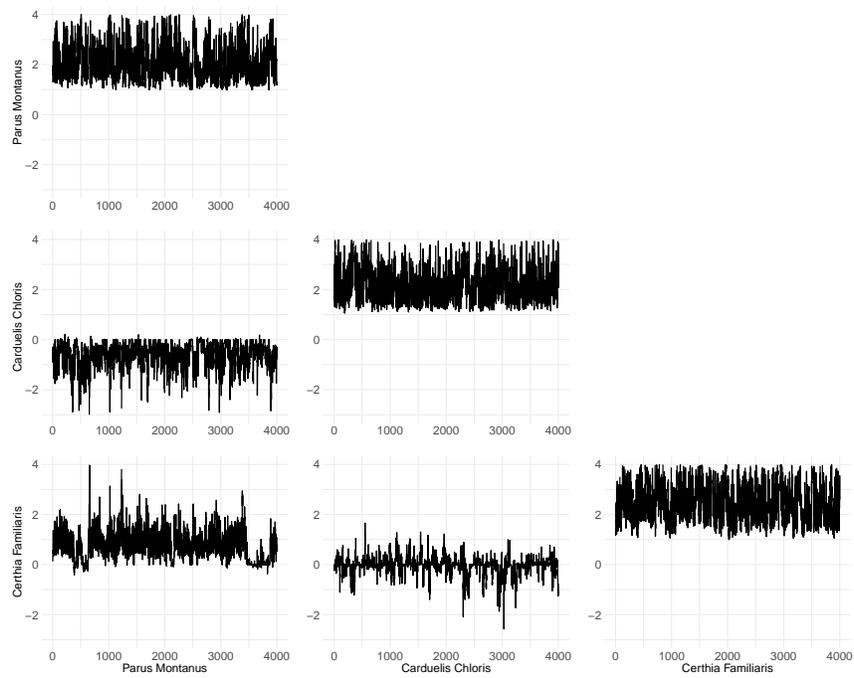}
	\caption{Chain plots of the marginal posterior samples of six elements of the covariance matrix obtained by the Gibbs sampler, discarding the first 20000 iterations and saving every $5$-th sample.}
	\label{fig:sm-mixing2}
\end{figure}

\clearpage

\subsection{Figures}
\begin{figure}[h]
	\centering
	\includegraphics[width=.85\textwidth]{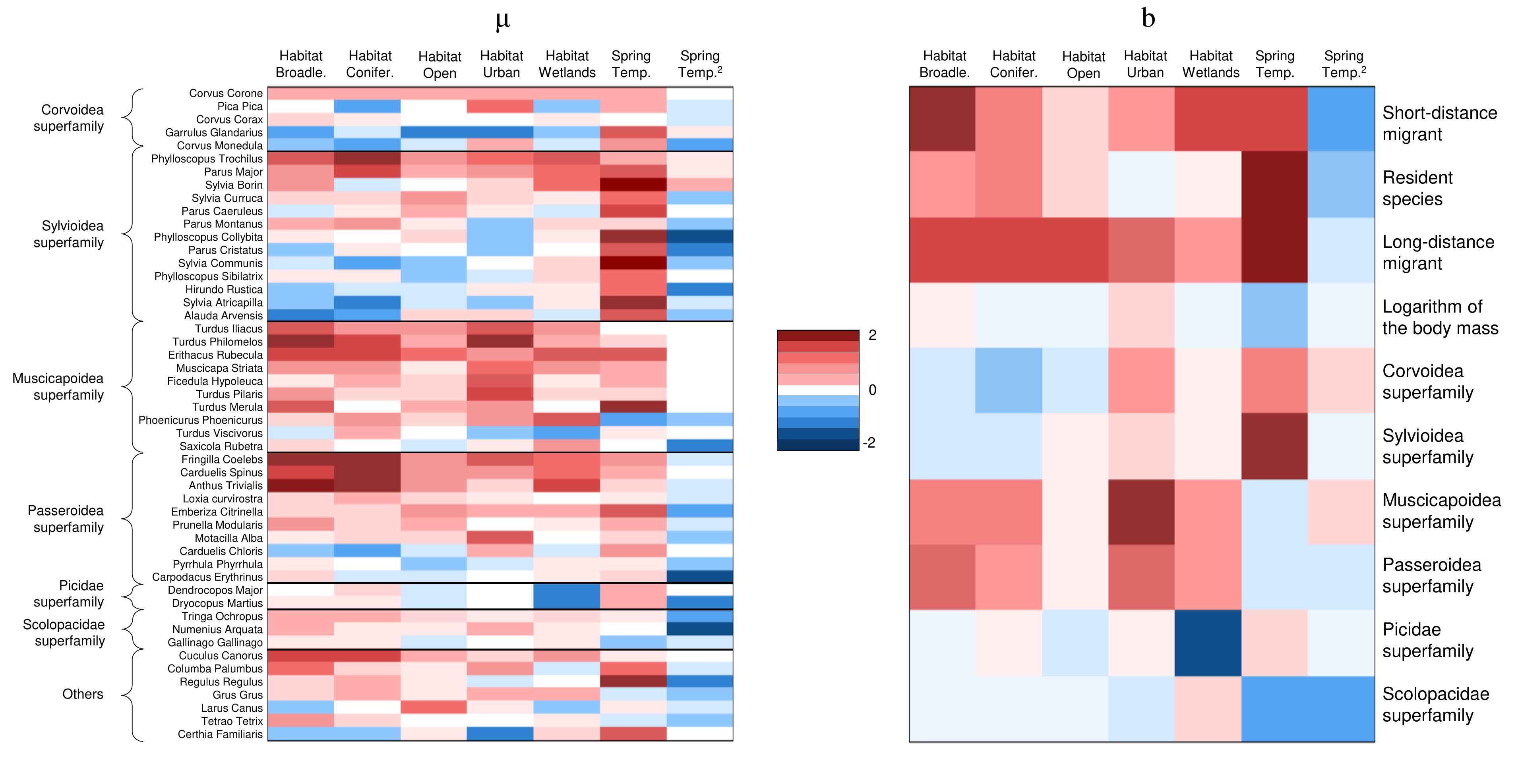}
	\caption{Posterior mean of $\mu$ and $b$ for the structured increasing shrinkage model; rows of left matrix refer to the 50 birds species, and rows of right matrix to the ten species traits. Broadle: broadleaved forests; Conifer: coniferous forests; Temp: temperature. }
	\label{fig:app-coeff}
\end{figure}

\begin{figure}[h]
	\centering
	\subfigure{\includegraphics[width=.28\textwidth]{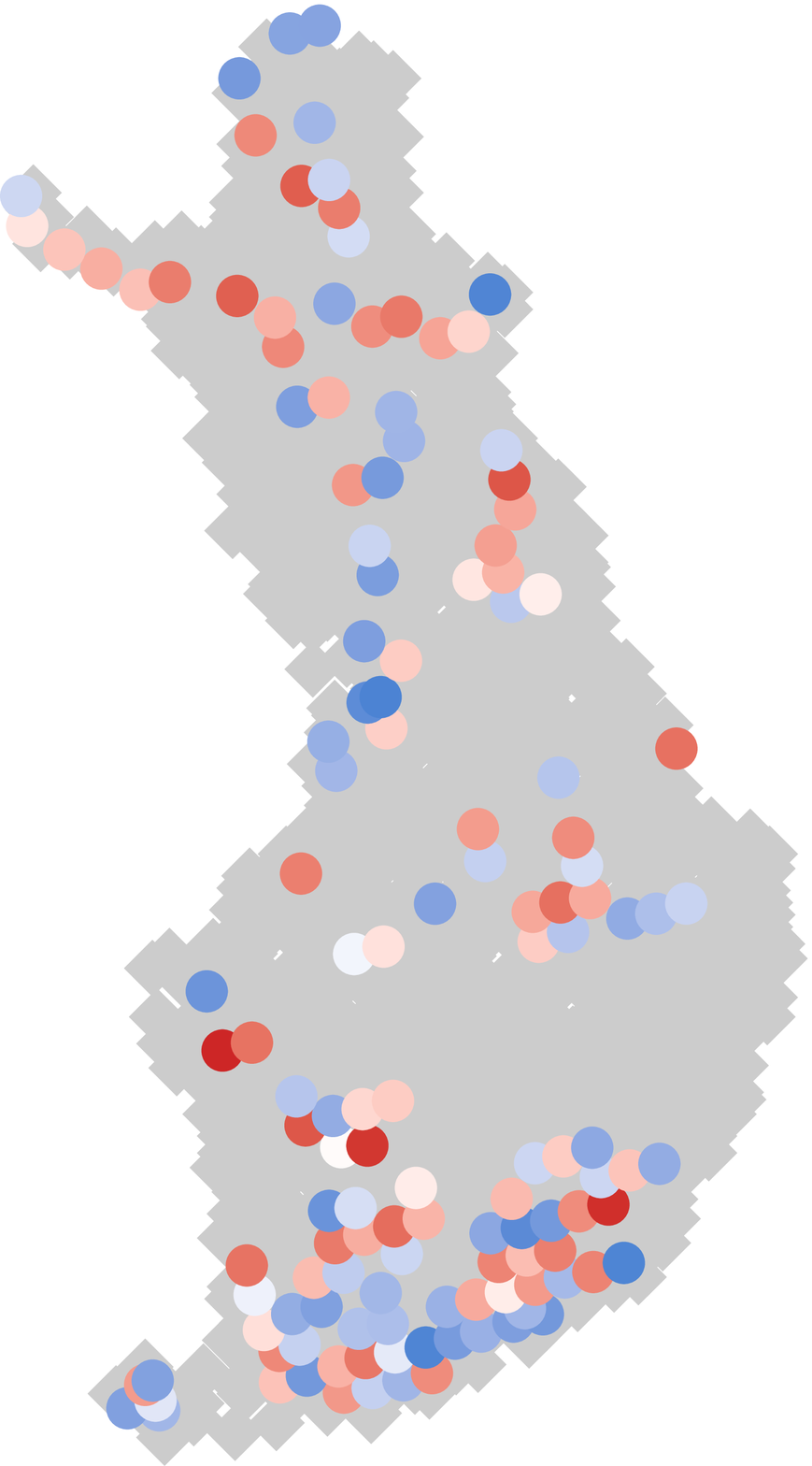}} 
	\hspace*{20pt}
	\subfigure{\includegraphics[width=.28\textwidth]{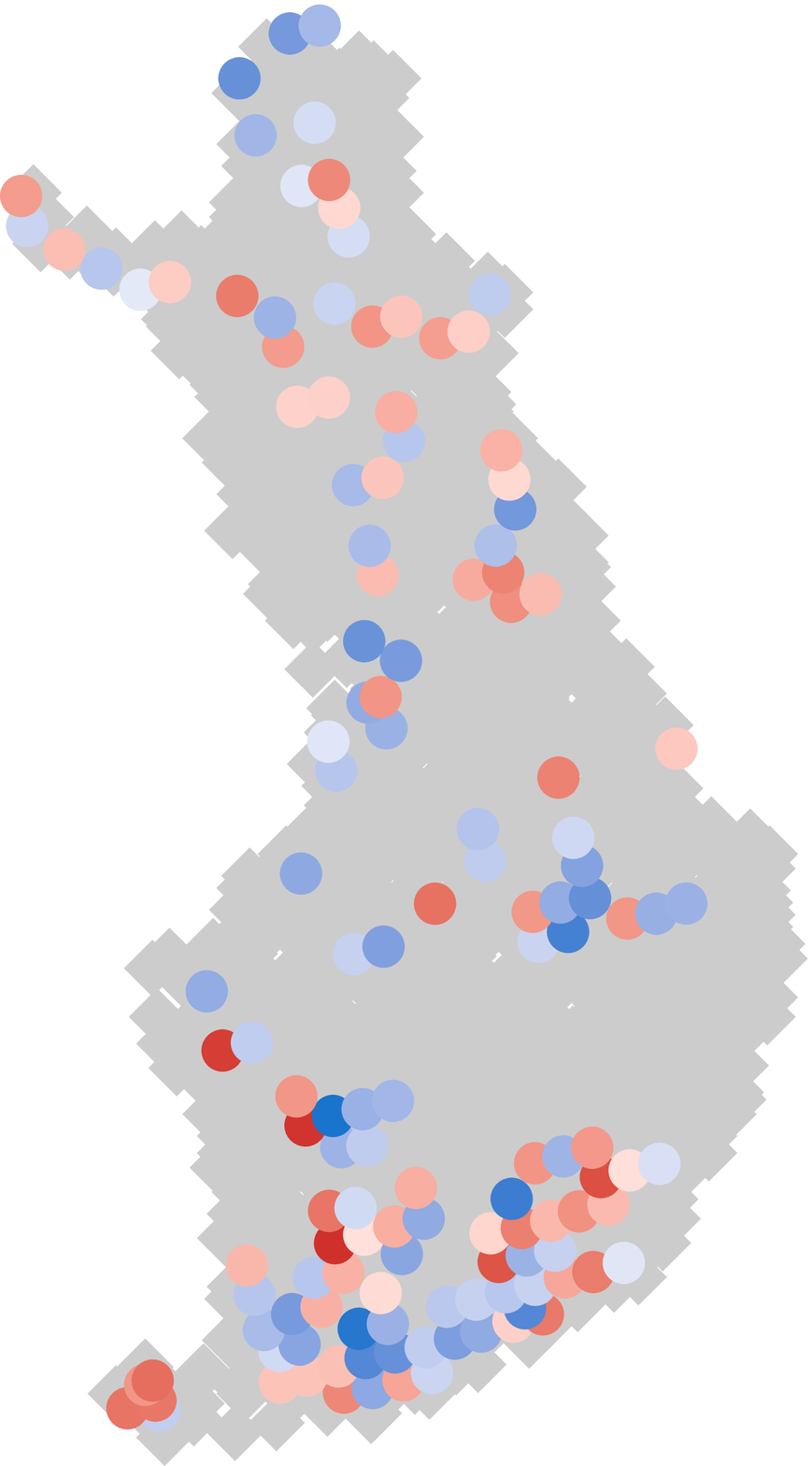}}
	\caption{Maps of the sampling units in Finland coloured accordingly to the values of the first and the third latent factors sampled at iteration $t^*$. Red and blue spots represent the environments with positive and negative values of the factors, respectively. }
	\label{fig:sm-app}
\end{figure}

\end{document}